\documentclass[prd,aps,amsfonts,eqsecnum,superscriptaddress,nofootinbib,notitlepage,longbibliography]{revtex4-1}
\usepackage{bbm,microtype,mathrsfs,amsmath,amssymb,color,amsthm,graphicx,bm}

\usepackage[colorlinks=true, urlcolor=violet, linkcolor=blue, citecolor=red, hyperindex=true, linktocpage=true]{hyperref}
\makeatletter
\renewcommand{\p@subsection}{}
\renewcommand{\p@subsubsection}{}
\makeatother

\usepackage{xcolor}
\usepackage{mathtools}

\usepackage{dsfont}

\usepackage[english]{babel}
\usepackage[utf8]{inputenc}
\usepackage[T1]{fontenc}

\usepackage{comment}
\usepackage[a4paper,top=2cm,bottom=2cm,left=2cm,right=2cm,marginparwidth=1cm]{geometry}

\usepackage{graphicx}
\usepackage[colorinlistoftodos]{todonotes}
\usepackage{tikz-cd}
\usepackage{xcolor}
\usepackage{braket}

\renewcommand{\vec}{\bm}

\newcommand{\CD}{\mathcal{D}}

\newcommand{\BE}{\mathbb{E}}
\newcommand{\CF}{\mathcal{F}}

\newcommand{\CL}{\mathcal{L}}
\newcommand{\CM }{\mathcal{M}}

\newcommand{\CO}{\mathcal{O}}

\newcommand{\BP}{\mathbb{P}}

\newcommand{\BR}{\mathbb{R}}

\newcommand{\CZ}{\mathcal{Z}}
\newcommand{\lV}{\lVert}
\newcommand{\e}{\mathrm{e}}
\newcommand{\rV}{\rVert}

\newcommand*{\tr}{\mathrm{Tr}}
\renewcommand{\L}{\left}
\newcommand{\R}{\right}

\newcommand{\vertiii}[1]{{\left\vert\kern-0.25ex\left\vert\kern-0.25ex\left\vert #1 \right\vert\kern-0.25ex\right\vert\kern-0.25ex\right\vert}}
\newcommand{\norm}[1]{\Vert {#1} \Vert}

\newcommand{\normp}[2]{\norm{#1}_{#2}}
\newcommand{\fnorm}[1]{\norm{#1}_{\mathrm{F}}}
\newcommand{\lnormp}[2]{\lnorm{#1}_{#2}}

\newcommand{\labs}[1]{\left\vert {#1} \right\vert}
\newcommand{\lnorm}[1]{\left\Vert {#1} \right\Vert}

\newcommand{\indicator}{\mathbb{I}}

\usepackage{amsthm}
\newtheorem{thm}{Theorem}
\numberwithin{thm}{section}
\newtheorem{cor}[thm]{Corollary}
\newtheorem{lem}[thm]{Lemma}

\newtheorem{prop}[thm]{Proposition}

\newtheorem{fact}{Fact}

\makeatletter
\renewcommand{\p@subsection}{}
\renewcommand{\p@subsubsection}{}
\makeatother

\begin{document}
\title{Optimal Frobenius light cone in spin chains with power-law interactions}

\author{Chi-Fang Chen}
\email{chifang@caltech.edu}
\affiliation{Institute for Quantum Information and Matter,
California Institute of Technology, Pasadena, CA, USA}

\author{Andrew Lucas}
\email{andrew.j.lucas@colorado.edu}
\affiliation{Department of Physics and Center for Theory of Quantum Matter, University of Colorado, Boulder CO 80309, USA}

\date{\today}

\begin{abstract}
     In many-body quantum systems with spatially local interactions, quantum information propagates with a finite velocity, reminiscent of the ``light cone" of relativity.  In systems with long-range interactions  which decay with distance $r$ as $1/r^\alpha$, however, there are multiple light cones which control different information theoretic tasks. We show an optimal (up to logarithms) ``Frobenius light cone" obeying $t\sim r^{\min(\alpha-1,1)}$ for $\alpha>1$ in one-dimensional power-law interacting systems with finite local dimension: this controls, among other physical properties, the butterfly velocity characterizing many-body chaos and operator growth. We construct an explicit random Hamiltonian protocol that saturates the bound and settles the optimal Frobenius light cone in one dimension.  We partially extend our constraints on the Frobenius light cone to a several operator $p$-norms, and show that Lieb-Robinson bounds can be saturated in at most an exponentially small $\e^{-\Omega(r)}$ fraction of the many-body Hilbert space.
     
\end{abstract}

\maketitle

\tableofcontents

\section{Introduction}

The celebrated Lieb-Robinson theorem proves that for arbitrary systems with nearest-neighbor interactions (with bounded strength), the speed at which quantum information can propagate is finite \cite{Lieb1972,Hastings_koma,Bravyi2006,chen2019operator, D_A_graph_GU_LUcas,PRXQuantum.1.010303}.  Intuitively, the time $t$ it takes to prepare an entangled state between two qubits (in an initially unentangled state) separated by distance $r$ is lower bounded as \begin{equation}
    t \ge \frac{r}{v_{\mathrm{LR}}}, \label{eq:babyLR}
\end{equation}
with $v_{\mathrm{LR}}$ the so-called Lieb-Robinson velocity.   In many respects, $v_{\mathrm{LR}}$ is analogous to the speed of light $c$ in special relativity -- no signals can be sent faster than $v_{\mathrm{LR}}$.  For this reason, we call (\ref{eq:babyLR}) the \emph{Lieb-Robinson light cone}.  Over the past few decades, numerous unexpected and important results have been shown to follow from the Lieb-Robinson theorem, including the exponential decay of correlation functions in a gapped ground state, together with an area law for entanglement \cite{Bravyi2006} and the effectiveness of matrix product state representations~\cite{PhysRevB.73.085115}, proofs of the ability to efficiently simulate local quantum systems \cite{haah2020quantum}, and demonstration of the stability of topological order~\cite{Bravyi_2010}. 


\subsection{Alternative choices of norms}



In quantum information dynamics, we study information propagation through the commutator between local operators (e.g. Pauli matrices) localized  on lattice sites $x$ and $y$.  This commutator takes the form of $[A_x(t),B_y]$, where 
\begin{equation}
    A_x(t) := \mathrm{e}^{\mathrm{i}Ht} A_x \mathrm{e}^{-\mathrm{i}Ht}
\end{equation}
denotes Heisenberg time evolution.\footnote{For simplicity, we assume in this equation that Hamiltonian $H$ does not depend on time. In all main results in this paper, we will eventually relax this assumption.} Indeed, local operators commute $[A_x(0),B_y]=0$, and a non-vanishing commutator indicates the operator $A_x(t)$ has ``grown" to a distant site $y$.  

The traditional Lieb-Robinson bounds quantify the growth of this commutator in the infinity norm (or the operator norm): $\lVert [A_x(t),B_y]\rVert_\infty$, being the maximal singular value of the operator $A$.  In slightly more physical terms, \begin{equation}
    \lVert A\rVert\equiv\lVert A\rVert_\infty = \sup_{\psi,\psi^\prime}\left| \langle \psi^\prime| A |\psi\rangle \right| \label{eq:operatornorm}
\end{equation}
is the maximal matrix element of $A$ between any normalized many-body states $\psi$ and $\psi^\prime$ and we drop the subscript. A bound on this operator norm then rules out the possibility of signaling, under the ``worst-case" choices of states $\psi$ and $\psi^\prime$. For quantum technology applications, such a bound is relevant to a device only when the \emph{many-body} initial and final states can be controlled exactly. If the device is imperfect and prone to some local errors, then we would like to know if other states in Hilbert space can also transmit information similarly fast. 

  Recently, \cite{hierarachy} emphasized that there can be \emph{multiple different light cones} for quantum information. Mathematically, these different light cones correspond to the choices of norms, which capture different physics and constrain different quantum information theoretic tasks; the physical question would be, ``what quantum mechanical task are we trying to bound?" 

As a concrete example, we consider whether or not one can signal between two qubits $x$ and $y$ \emph{independently} of the actual quantum state.  This amounts to averaging over the magnitude of $\langle \psi^\prime| A |\psi\rangle$ over all pairs of $\psi$ and $\psi^\prime$~\cite{hierarachy}. This can be captured by the \emph{square} of matrix element's magnitude, and (for an operator $A$) we then recognize that this average is nothing but the normalized 2-norm, which will be called the \emph{Frobenius norm} \begin{equation}
    \lVert A \rVert_{\mathrm{F}}^2 := \frac{\mathrm{Tr}(A^\dagger A)}{\mathrm{Tr}(I)}. \label{eq:frobeniusnorm}
\end{equation}
Here $I$ denotes the identity matrix.  If we wish to design a quantum device in which signaling or quantum state transfer is achieved regardless of the details of the many-body state, we will care not about the Lieb-Robinson bound on a commutator, but instead a ``Frobenius bound" on $\lVert [A_x(t),B_y]\rVert_{\mathrm{F}}$.
  
  Interestingly, this Frobenius norm has a natural interpretation as an (out-of-time-ordered) correlator at infinite temperature
 \begin{equation}
    \lVert [A_x(t),B_y] \rVert^2_{\mathrm{F}} = \frac{\mathrm{Tr}\left([A_x(t),B_y]^\dagger [A_x(t),B_y] \right)}{\mathrm{Tr}(I)} = \mathrm{Tr}\left(\rho_{T=\infty}  [A_x(t),B_y]^\dagger [A_x(t),B_y]\right).
\end{equation} 
Here $\rho_{T=\infty}\propto I$ denotes the infinite temperature density matrix. Such objects have been used to diagnose many-body chaos in extensive studies over the past few years. Indeed, when the Frobenius norm becomes large for any operator $B_y$, information about the initial operator $A_x$ cannot be retrieved via local measurements. 




A simple way to define a ``light cone" out of either of these norms is to define $t(r)$ as the first time that qubits separated by distance $r$ have a large commutator: e.g., \begin{equation}
  t(r) := \mathrm{arginf}_t \left\lbrace 0<t<\infty :  \lVert [A_0(t),B_r]\rVert_* \ge \frac{1}{2} \right\rbrace
\end{equation}
The notation here is simply that we look for the smallest time (arginf) at which a certain commutator norm is bigger than some $O(1)$ number, such as $\frac{1}{2}$. It is interesting to ask whether in \emph{all} Hamiltonian dynamics consistent with various simple constraints (e.g. nearest-neighbor interactions) there are bounds on $t(r)$.  For the operator norm, the answer is provided by the Lieb-Robinson theorem (\ref{eq:babyLR}).  But, is it possible that in typical states (Frobenius norm)\footnote{Sometimes the Frobenius light cone $v_{\mathrm{F}}\sim r/t(r)_{\mathrm{F}}$ is called butterfly velocity in the quantum chaos literature.}, the light cone is slower:  
\begin{align}
v_{\mathrm{F}}<v_{\mathrm{LR}}?    
\end{align}

Preliminary work \cite{Nahum:2017yvy,vonKeyserlingk:2017dyr} noticed that in random unitary circuits, this was indeed the case.  In other words, the choice of norm (or, more physically, the task of interest) quantitatively changed the ``speed limits" on quantum dynamics.   

With local interactions, this change is typically quantitative (but not qualitative).  With long-range interactions, however, the discrepancy between Lieb-Robinson and Frobenius bounds becomes drastic. This perhaps explains why the Frobenius norm did not draw serious attention in the study of locally interacting systems.

\subsection{Long-range interactions}
We now turn to physical systems with power-law, or long-range, interactions.
Many promising platforms for future quantum technologies do not merely contain nearest neighbor interactions: they can include trapped ion crystals \cite{Britton2012}, Rydberg atom arrays \cite{Saffman2010}, polar molecules \cite{Yan2013}, etc.; each of these platforms has long-range interactions between all pairs of qubits which decay with the distance $r$ as $V(r)\propto r^{-\alpha}$.  Here $\alpha$ is an exponent characterizing the system:  $0<\alpha<3$ can approximately be achieved in trapped ion crystals, while $\alpha=3$ for polar molecules with dipolar interactions and $\alpha=3,6$ for Rydberg atoms.  The classic techniques from Lieb-Robinson bounds were ineffective at constraining information dynamics in long-range systems. Over the past few years, increasingly sophisticated methods have been developed \cite{Hastings_koma,GongFF,Foss-FeigG,elseImprovedLiebRobinsonBound2018,Tran_2019_polyLC,alpha_3_chenlucas,kuwaharaStrictlyLinearLight2020,2020arXiv201002930T,tran2021optimal} to conclusively settle  the optimality of Lieb-Robinson-like bounds.  In a $d$-dimensional system, the time it takes to prepare an EPR state between two qubits, separated by distance $r$, scales as $t \sim r^{\min(\alpha-2d,1)}$ for $\alpha>2d$.  This remarkable result means that even when $\alpha=6$, in three dimensions it is possible to prepare highly entangled states in nearly constant time ($t\sim \log r$).

It was first noted in \cite{hierarachy} that the Lieb-Robinson bounds did not effectively constrain the Frobenius norm of a commutator in models with power-law interactions.  There, it was argued that the bound on Frobenius norms would scale as \begin{equation}
    t_{\mathrm{F}}^*(r) \gtrsim r^{\min(\alpha-3d/2,1)}, \label{eq:oldfrob}
\end{equation}
 which was proved in $d=1$ and also shown to be ``optimal".  However, strictly speaking, the OTOCs which saturate (\ref{eq:oldfrob}) are \emph{not} of the form $\lVert[A_x(t),B_y]\rVert_{\mathrm{F}}$, with $A_x$ and $B_y$ simple Pauli matrices.  Instead, they are of the form (in $d=1$, for simplicity): \begin{equation*}
    \left\lVert \left[ A_x(t), \prod_{z\ge y} B^\prime_z \right]\right\rVert_{\mathrm{F}},
\end{equation*}
with $B^\prime_z$ Pauli raising operators on sites $z \ge y$.\footnote{ Products of raising operators has low rank and a very large singular value, unlike simple Pauli.}


However, (\ref{eq:oldfrob}) is not optimal if we demand $A_x$ and $B_y$ are both local matrices.
It was recently proved in any dimension $d$ that for $\alpha>d$, the Frobenius light cone for such an OTOC is of the form \cite{kuwahara2021} \begin{equation}
    t \gtrsim r^{2(\alpha-d)/(2\alpha-d+1)}.
\end{equation}

\subsection{Our results}
The main motivation of this paper is to reconcile the results of \cite{hierarachy,kuwahara2021}, and prove an optimal Frobenius bound of small operators in systems with long-range interactions, in one dimensional models $(d=1)$.  We find in Section \ref{sec:mainthm} that \begin{equation}
    t_{\mathrm{F}}(r) \gtrsim r^{\min(\alpha-1,1)}, \;\;\; (\alpha>1), \label{eq:tF}
\end{equation}
where logarithms and absolute constants(that only depends on $\alpha$) are suppressed. In Section \ref{sec:protocol}, we construct an explicit random dynamical protocol that saturates this scaling (up to factors of $\log r$). 

We conjecture that in $d$ spatial dimensions, (\ref{eq:tF}) generalizes to \begin{equation}
    t_{\mathrm{F}}(r) \sim r^{\min(\alpha-d,1)}
\end{equation} for $\alpha>d$.  In fact, we will present a random protocol which saturates such a bound.  Our bounds together with a saturating protocol definitively settle the nature of the Frobenius light cone in $d=1$ spatial dimensions.  For technical reasons we explain in the body of the work, proving Frobenius light cones in higher dimensions is much more difficult and thus remains a conjecture for $d>1$.

In our proof, we employ a new set of mathematical tools for deriving bounds on quantum dynamics with other norms in Section \ref{sec:submulti}.  In a nutshell, the standard Lieb-Robinson bounds work by exploiting submultiplicativity, which states that \begin{equation}
    \lVert A B \rVert \le \lVert A \rVert \cdot \lVert B \rVert. \label{eq:submult}
\end{equation}
The traditional Lieb-Robinson bound strategy proceeds by iteratively using this identity repeatedly to a nested commutator of the form $\lVert [H,[H,\cdots \mathcal{O}]]\rVert t^n/n!$, reducing the problem of bounding quantum dynamics to a combinatorial problem.  This property does \emph{not} hold for the normalized Frobenius norm introduced in (\ref{eq:frobeniusnorm}). As a consequence, existing proofs \cite{hierarachy,Lucas:2019cxr,Yin:2020pjd,Yin:2020oze} of Frobenius light cones have typically involved other methods inspired by quantum walks and probability theory: these methods are, in turn, uniquely adapted to the Frobenius norm.  A key development in this paper is to show that (\ref{eq:submult}) can be \emph{nearly} generalized to the \emph{normalized} Frobenius norm (and also Schatten p-norms below), when the operators $A$ and $B$ obey certain spatial locality constraints: \begin{equation}
    \lVert A B \rVert_{\mathrm{F}} \lesssim \lVert A \rVert_{\mathrm{F}} \cdot \lVert B\rVert_{\mathrm{F}} \ln \lVert B\rVert_{\mathrm{F}}.\label{eq:F_submulti}
\end{equation}
Our results therefore allow us to use traditional Lieb-Robinson-like techniques, but to obtain \emph{parametrically stronger bounds} on dynamics using more physically relevant norms.

We are also able to obtain results for $p$-norms which interpolate between the Frobenius and operator norms defined above.  We generalize the proof of \cite{alpha_3_chenlucas} to any Schatten $p$-norms $[\mathrm{Tr}([A_x(t),B_y]^{p})/\mathrm{Tr}(I)]^{1/p}$.  We prove that there is a \textit{p-dependent} light cone of the form 
\begin{align}
t(r)_p\gtrsim \frac{ r^{\min(1,\alpha-3/2)}}{\sqrt{p}}
\end{align}
for any fixed value of $p$. Intuitively, these bounds constrain the $p$-th \emph{moments} of the singular value distribution of $[A_x(t),B_y]$, and therefore tell us how rare extreme singular values are. While we don't believe this bound is tight, it is already tighter than Lieb-Robinson bounds. In particular, we will show that Markov's inequality implies it is exponentially rare to find a state in Hilbert space that can saturate a Lieb-Robinson light cone. 


\subsection{Implications of our results and further conjectures}
Our results have a number of interesting implications.   In systems with power-law interactions, we find a remarkable result: in a spin chain with $\alpha=2$ the growth of operators is hardly faster than with nearest neighbor interactions, while existing state transfer protocol may send a single qubit in time $t\sim \log r$ \cite{2020arXiv201002930T}! 
At the same time, the rapid state transfer algorithm of \cite{2020arXiv201002930T} requires an \textit{exponentially} fine-tuned background state to implement in quantum processors. 
Our results therefore have important practical implications for the wide range of possible quantum technologies, such as trapped ion crystals, which might seek to use power-law interactions to speed up the transfer of information.  It will be interesting to understand the actual performance of the protocol of \cite{2020arXiv201002930T} in the presence of some errors.

We conjecture that the Frobenius light cone is an \emph{upper bound} on the time it takes to generate a volume-law entangled state.  Intuitively, generating volume law entanglement requires taking all local information and scrambling it amongst highly non-local degrees of freedom; thus, a local operator must evolve into a non-local one.  Since operator growth is captured by the Frobenius norm, the Frobenius light cone should also control entanglement.  Support for this conjecture is found in the examples of \cite{shor}.  

While we believe this intuitive argument to be correct, we have not been able to formalize a proof of this result.  Nevertheless, if this conjecture is correct, it would immediately imply a very surprising result: namely, the (asymptotically) optimal way to send volume law entangled states using power-law interactions is to send one qubit at a time.  Indeed, the protocol of \cite{2020arXiv201002930T} shows that we can send one qubit in a time $r^{\alpha-2d}$ for $2d<\alpha<2d+1$.  However, if volume law entanglement can only be generated in time $r^{\alpha-d}$, that means that sending one qubit at a time using the fast protocol of \cite{2020arXiv201002930T} takes a time $r^{\alpha-2d}\times r^d \sim r^{\alpha-d}$, which is the fastest possible.  This observation can simplify the design of any future quantum processor performing fast quantum state transfer with power law interactions.

Another implication of our results is that the Frobenius bound we present is essentially optimal even when the Hamiltonian is drawn randomly - in fact our saturating protocol is designed using an ensemble of saturating Hamiltonians, which are themselves random. In contrast, typical random Hamiltonians have sharper operator norm bounds than deterministic Hamiltonian~\cite{chen2021concentration}. This demonstrates that the Frobenius bounds is more robust than the operator norm, in that much less fine tuning of Hamilotonians is needed to saturate a Frobenius bound.  In this sense, Frobenius bounds are slightly more universal than Lieb-Robinson bounds.

Our mathematical results may help to improve existing ``many-body quantum walk" techniques for bounding OTOCs in a variety of models \cite{Lucas:2019cxr,Lucas:2020pgj,Yin:2020oze,Yin:2020pjd,Chen:2020bmq}.  A particularly challenging problem where our methods may be extremely valuable is bounding finite temperature correlators, where the appropriate ``norm" depends on the thermal density matrix \cite{lucasprl19}.  Understanding how to derive specific notions of submultiplicativity, which might hold under this norm, but not in general, may help to solve longstanding conjectures about the speed limits on finite temperature dynamics which have, in recent years, linked developments in string theory and quantum gravity to many-body physics \cite{Maldacena:2015waa,blake,swingle}. 

Lastly, our bounds on Frobenius light cones may help to find sharper error bounds on new numerical techniques such as matrix product operator methods \cite{Pirvu_2010}, since so long as a thermal correlator is being studied, one is not interested in the error in the numerical evaluation of an operator in the worst-case state, but only in a typical state in the ensemble.




\section{Preliminaries}
In this section we review formalisms in operator dynamics and useful facts from non-commutative functional analysis.
\subsection{Long-range interactions}\label{sec:preLR}
We first provide a careful definition for long-range interacting systems. 
Let $\Lambda$ denote the vertices of a lattice graph, consisting of some unit cell repeated periodically in $d$ spatial dimensions.\footnote{Mathematically, if $E_\Lambda$ is an edge set on $\Lambda$, we demand that the pair $(\Lambda,E_\Lambda)$ has an automorphism group containing a translation subgroup $\mathbb{Z}^d$.} For vertices $i,j\in \Lambda$, let us define the distance $d(i,j)$ between the two vertices to be the Manhattan distance (minimal number of edges to traverse to get from $i$ to $j$).  For ease of presentation, let us consider a many-body Hibert space consisting of a single qubit (two-level system) on every lattice site.  The Hilbert space of operators acting on a single qubit on site $x$ is spanned by the identity $I_x$, and the three Pauli matrices $X_x,Y_x,Z_x$.  In what follows, we denote with $X^a = (X,Y,Z)$ the set of all non-trivial Paulis.

Given a 2-local Hamiltonian \begin{equation}
    H(t) = \sum_{x,y\in \Lambda} \sum_{a,b=1}^3 J_{xy}^{ab}(t) X_x^a X_y^b,
\end{equation} 
we say that it has power-law interactions of exponent $\alpha$ if \begin{equation}
    |J_{xy}^{ab}(t)| \le \frac{1}{d(x,y)^\alpha}. \label{eq:Jab}
\end{equation}
Without loss of generality, we set the prefactor above to be 1 by rescaling the Hamiltonian. In practice, it is useful to say a model has exponent $\alpha$ if it does not have exponent $\alpha^\prime>\alpha$, in our bounds the only criterion necessary is (\ref{eq:Jab}).

\subsection{Operator norms and operator size}
In this paper, we are interested in the Heisenberg time evolution of an operator \begin{equation}
    \frac{\mathrm{d}A(t)}{\mathrm{d}t} := \mathrm{i}[H(t),A] := \mathcal{L}(t)[A(t)].
\end{equation}
It will be helpful to interpret operator dynamics in a bra-ket notation on the ``Hilbert space" of operators, in which case we write the above equation as  \begin{equation}
\frac{\mathrm{d}}{\mathrm{d}t} |A(t)) := \mathcal{L}(t)|A(t)).
\end{equation}
We define a Frobenius inner product on this space.  Loosely speaking\footnote{On an infinite lattice, this ratio is not well-defined.  We may interpret it so long as the operators $A$ and $B$ have compact support, in which case the traces may be restricted to the support of $A$ and $B$.   Taking such an operator $A$ with compact support, we may also time-evolve it.  With a Hamiltonian $H$ without compact support, the operator $A(t)$ will (in general) also no longer have compact support.  In this case, we may define $(A|B)$ via a limiting process: letting $A_S$ denote the terms in $A$ supported on compact set $S$,  the limit of $(A_S|B_S)$ will converge as $S$ grows; the convergence is guaranteed by Lieb-Robinson theorems. }, \begin{equation}
    (A|B) := \frac{\mathrm{tr}(A^\dagger B)}{\mathrm{tr}(I)}. \label{eq:innerproduct}
\end{equation}
It is often helpful to define projection operators.   For example, the superoperator $\mathbb{P}_x$, defined via \begin{equation}
    \mathbb{P}_x A = -\frac{1}{8}\sum_{a=1}^3 [X_x^a,[X_x^a,A]],
\end{equation}
allows us to restrict to all terms in $A$ which act non-trivially on site $x$. 
Since \begin{equation}
    \lVert [B_y,A]\rVert_{\mathrm{F}} = \lVert [B_y,\mathbb{P}_yA]\rVert_{\mathrm{F}} \le 2 \lVert B_y\rVert_\infty \lVert \mathbb{P}_y A\rVert_{\mathrm{F}} = 2 \lVert B_y\rVert_\infty \sqrt{(A|\mathbb{P}_y|A)}, \label{eq:infF}
\end{equation}
it suffices to bound the (Frobenius) norm of $\mathbb{P}_y|A)$, in order to bound OTOCs.

In fact, without loss of generality, we may write any operator $A$ as \begin{equation}
    A = \sum_{S\subset \Lambda} A_S(t), \;\;\;\; \text{where } \mathbb{P}_x A_S(t) = \left\lbrace\begin{array}{ll} A_S(t) &\ x\in S \\ 0 &\ \text{otherwise} \end{array}\right..  \label{eq:AS}
\end{equation}

Assuming $(A|A)=1$ is normalized, we will then refer to \begin{equation}
    p_S = (A_S|A_S) \label{eq:pS}
\end{equation} as the probability that the operator is supported on subset $S$.   Clearly, \begin{equation}
    \sum_S p_S=1,
\end{equation} so this terminology is well-defined.  This probabilistic language will be highly valuable for us as we develop a fast operator growth protocol.

While the most important operator norms in this paper are the Frobenius norm (\ref{eq:frobeniusnorm}), along with the operator norm (\ref{eq:operatornorm}), we will also find useful the Schatten $p$-norms\footnote{As before, one should rigorously only apply this definition on a finite dimensional Hilbert space.  We will implicitly invoke the limiting process explained above in this paper, in order to make sense out of the $p$-norm on an infinite dimensional space.}  that interpolate in between the Frobenius and operator norms:
\begin{equation}
    \lV X\rV_p :=\tr [(X^\dagger X)^{p/2} ]^{1/p}.
\end{equation}
We will often deal with the normalized p-norm divided by $\mathrm{tr}(I)$ that
\begin{align}
   \normp{X}{\bar{p}}\le  \frac{\lV X\rV_p}{\normp{I}{p}}\le \lV X\rV_\infty \le \lV X\rV_p.
\end{align}
To show an approximate form of submultiplicativity for the Frobenius norm, we will also find useful the following standard inequality from functional analysis: 
\begin{prop}[non-commutative H\"older's inequality \cite{chen2021concentration}] \label{prop:holder}
If \begin{equation}
    \frac{1}{p} = \frac{1}{p_1}+\frac{1}{p_2},
\end{equation} then \begin{equation}
    \lVert AB\rVert_{p} \le \lVert A\rVert_{p_1} \lVert B\rVert_{p_2}
\end{equation}
\end{prop}
Additionally, we have:
\begin{prop}[Riesz-Thorin interpolation] \label{prop:riesz}
If \begin{equation} 
    \frac{1}{q_\theta} = \frac{\theta}{q_1} + \frac{1-\theta}{q_2},
\end{equation} then \begin{equation}
    \lVert A\rVert_{q_\theta} \le \lVert A\rVert_{q_1}^\theta \lVert A\rVert_{q_2}^{1-\theta}.
\end{equation}
\end{prop}

\section{Towards submultiplicativity for the Frobenius norm}\label{sec:submulti}
In this section, we develop a mathematical machinery which will be critical to proving the Frobenius light cone, and its $p$-norm extensions.  In the main text, we will focus on the simpler case of two-body (2-local) Hamiltonians; the straightforward extension to $k$-local Hamiltonians is provided in Appendix \ref{app:klocal}.

The following lemma represents the key technical result of this section:
\begin{lem}\label{lem:submulti}
Consider the 2-local Hamiltonian\begin{equation}
    H = \sum_{i<j} H_{ij}
\end{equation}
with each $H_{ij}$ supported on exactly sites $i$ and $j$ ($\mathbb{P}_k H_{ij} = \mathbb{I}(k\in\lbrace i,j\rbrace) H_{ij}$),  and $\lV O\rV \le 1$, then
\begin{align}
    \lV HO\rV_{\mathrm{F}} \le  2\e \lV H \rV_{(2)} \lV O\rV_{\mathrm{F}}\bigg(|\ln\lV O\rV_{\mathrm{F}}|+1\bigg). \label{eq:3prop}
\end{align}
where the Frobenius norm is normalized as in (\ref{eq:frobeniusnorm}), and 
\begin{align}
    \lV H \rV_{(2)}:= \sqrt{\sum_{i<j}  \normp{H_{ij}}{\infty}^2}.
\end{align}
\end{lem}
Before diving into the proof, let us put this result into some context and explain why this refinement is valuable. For simplicity, consider a operator $O$ of diameter $r$. As this paper is motivated by long-range interacting systems, we consider the effect of long-range interactions with another sphere of diameter $r$, a distance $O(r)$ away: see Figure \ref{fig:OH}. The leading order type of operator growth is
\begin{align}
    \mathrm{e}^{\mathrm{i}Ht}Oe^{-\mathrm{i}Ht} \approx O+ \mathrm{i}t[H,O] + \cdots,
\end{align}
and an unconditional bound would be 
\begin{align}
    \fnorm{HO} \le  \normp{H}{\infty} \fnorm{O}. \label{eq:normworst}
\end{align}
If we simply wish to bound $\fnorm{O}$, this triggers an recursive identity; however, it comes at the price of a large prefactor of $\normp{H}{\infty}$.  Indeed, observe that between two balls of size $r$ a distance $r$ apart, \cite{hierarachy} \begin{subequations}
    \label{eq:fig1eq}
    \begin{align}
        \lVert H\rVert_\infty &\lesssim \sum_{x\in \text{ball 1}} \sum_{y\in \text{ball 2}} \frac{1}{d(x,y)^\alpha} \sim \frac{1}{r^{\alpha-2d}}, \\
        \normp{H}{(2)}\sim \lVert H\rVert_{\mathrm{F}}  =\frac{\lV H \rV_2}{\normp{I}{2}} &\lesssim \sqrt{\sum_{x\in \text{ball 1}} \sum_{y\in \text{ball 2}} \left(\frac{1}{d(x,y)^\alpha}\right)^2} \sim \frac{1}{r^{\alpha-d}}. \label{eq:ball2}
    \end{align}
\end{subequations}
This equation suggests it is highly desirable to use $\lVert H\rVert_{(2)}$, \emph{not} $\lVert H\rVert_\infty$, if at all possible. Yet (\ref{eq:normworst}) suggests we must always use the worst case scenario, set by the operator norm.\footnote{Note that the normalization is absolutely crucial, otherwise the Schatten 2-norm is even looser than the operator norm, and is already submultiplicative:$
     \lV X\rV_2/\normp{I}{2}\le \lV X\rV_\infty \le \lV X\rV_2$.
}

An observation in \cite{kuwahara2021} was that the Frobenius norm can behave submultiplicatively when the other operator has bounded spectrum: as in (\ref{eq:infF}),
\begin{align}
    \lV HO\rV_{\mathrm{F}} \le  \lV H \rV_{\mathrm{F}} \lV O\rV_{\infty}. \label{eq:only_once}
\end{align}
If $\lV O\rV_{\infty} \approx \lV O\rV_{\mathrm{F}}$, as is the case when $O$ is a simple Pauli matrix ($O^2=1$), then using this inequality gives us an approximate form of submultiplicativity, one time.  Of course, clearly this does not feed in to a recursive bound; indeed in \cite{kuwahara2021} this trick was only used once.  But our refined bound is written entirely in terms of the Frobenius-like-norm $\normp{H}{(2)}$ (except for the global constraint $\lV O\rV\le 1$ inherited from the initial operator). This is more powerful and allows us to obtain much stronger Frobenius light cone bounds. We now present the proof of the lemma.

\begin{figure}[t]
    \centering
    \includegraphics[width=0.6\textwidth]{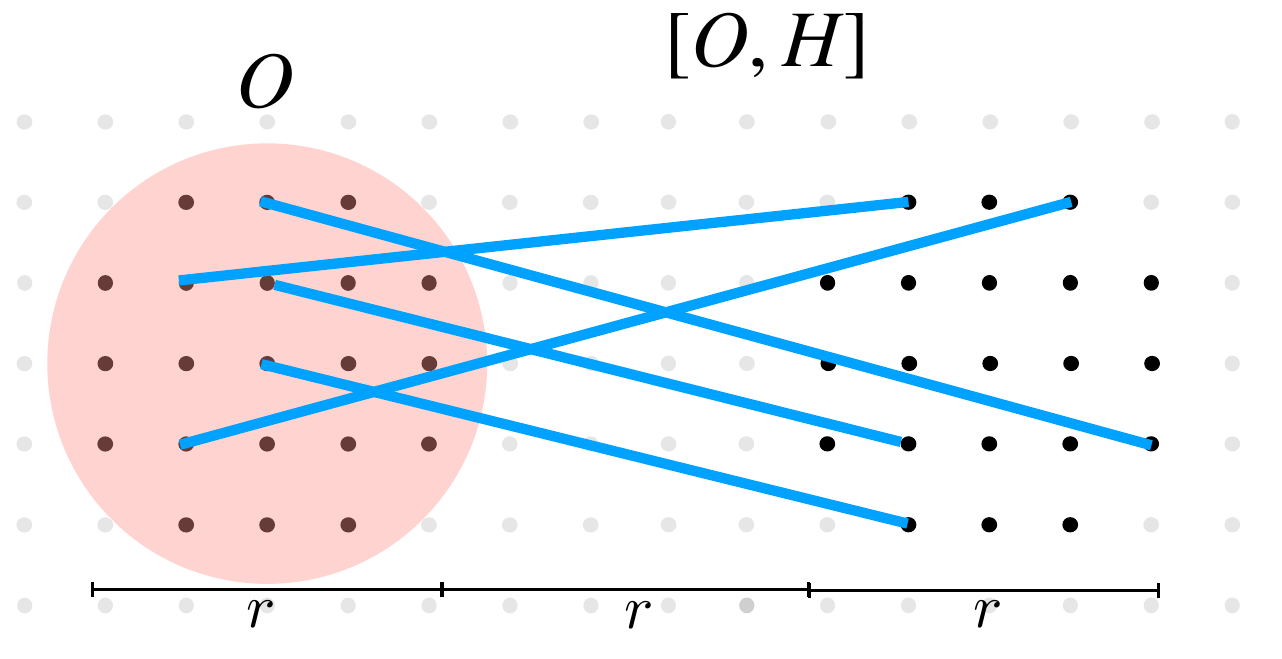}
    \caption{ Long-range interactions acting on a large operator between two balls of radius $r$. The interaction Hamiltonian between these two balls has norms given in (\ref{eq:fig1eq}).  The key technical advance of this paper is to learn how to take advantage of the smaller Frobenius norm of an interaction Hamiltonian in controlling the Frobenius light cone.
    }
    \label{fig:OH}
\end{figure}

\begin{proof}
Using Proposition \ref{prop:holder}, \begin{equation}
    \lVert HO\rVert_{2} =  \lVert HO\rVert_2 \le \lVert H\rVert_{2/(1-\theta)} \cdot \lVert O\rVert_{2/\theta}. \label{eq:p_holder3}
\end{equation}
We now bound each of these norms in turn. The idea is to choose $\theta\rightarrow 1$, without making the norm of $\normp{H}{2/(1-\theta)}$ too large.  Hence, we need to prove that the $p$-norms $\normp{H}{p}/\normp{I}{p}$ grow sufficiently slowly (when normalized by the Hilbert space dimension).  In other words, we need to prove bounds on the tails of the singular distribution of $H$.  In the mathematics/probability literature, this is called a concentration bound. 

More precisely, to bound $\normp{H}{q}$, $q:=2/(1-\theta)$,
 we will use a seemingly simple recursive moment inequality, proved in Appendix \ref{app:klocal}:
\begin{prop}[Uniform smoothness for subsystem]\label{prop:unif_subsystem}
Consider matrices $X, Y$ of the same dimensions that satisfy
$\tr_i(Y) = 0; X= X_j\otimes I_i$. For $q \ge 2$, 
\begin{equation}
\lVert X + Y\rVert_{q}^2\le \lVert X \rVert_{q}^2  + C_q\lVert Y\rVert_{q}^2 .
\end{equation}
for optimal constant $C_q=q-1$.
\end{prop}
Hence,
\begin{align}
    \normp{\sum_{i<j}  H_{ij}}{q}^2 &\le  \normp{\sum_{1<i<j }  H_{1j}}{q}^2 + C_q \normp{\sum_{1=i<j}  H_{ij}}{q}^2 \notag \\
    &\le \sum_{i} C_q \normp{\sum_{i<j}  H_{ij}}{q}^2\notag \\
    &\le \sum_{i<j} (C_q)^2 \normp{  H_{ij}}{q}^2 \notag \\
    &\le (C_q)^2\normp{I}{q}^2\sum_{i<j}  \normp{  H_{ij}}{\infty}^2 = (C_q)^2\normp{I}{q}^2 \normp{H}{(2)}^2 \label{eq:F_Hnorm}
\end{align}
In the first line we applied Proposition~\ref{prop:unif_subsystem} for $Y=\sum_{1=i<j}  H_{ij}$ and $X=\sum_{1<i<j }  H_{1j}$, i.e. ``peeling off" any term traceless on the qudit $i=1$. In the second inequality we peel off qudit $i=2,3,\ldots$ analogously. In the third line, we recursively repeat the first and second line for the $j$ index, which gives a second factor of $C_q$. Lastly, we use Holder to bound the q-norm by the $\infty$ norm with a $\normp{I}{q}$ overhead.\footnote{This also has the name of hypercontractivity that was proven in some special cases~\cite{montanaro_q_boolean,Montanaro_application_hypercontract}. We re-derive it using uniform smoothness which also has wider applications. See ~\cite{chen2021trotter} for a further discussion.}

(2) For the norm of $O$, a standard manipulation using Proposition \ref{prop:riesz}, and $\lVert O \rVert_\infty \le 1$, allows us to write \begin{equation}
    \lVert O\rVert_{2/\theta} \le \lVert O\rVert_2^\theta \cdot \lVert O\rVert_\infty^{1-\theta}  \le \lVert O\rVert_{2}^\theta  \label{eq:F_Onorm}
\end{equation}
Now combining (\ref{eq:frobeniusnorm}), (\ref{eq:p_holder3}), (\ref{eq:F_Hnorm}) and (\ref{eq:F_Onorm}), we plug in $q=2/(1-\theta)$ to get \begin{align}
    \fnorm{HO} =\frac{1}{\normp{I}{2}}\lVert HO\rVert_{2} &\le 
    \frac{1}{\normp{I}{2}} \left(\left(\frac{2}{1-\theta}-1\right) \normp{H}{(2)}\normp{I}{2/(1-\theta)}\right)\lVert O\rVert_{2}^\theta \le \left(\frac{2}{1-\theta}\right)\lVert H\rVert_{(2)}\left(\frac{\lVert O\rVert_{2}}{\normp{I}{2}}\right)^\theta
\end{align}
where we explicit display the $\normp{I}{2}$ and used that $\normp{I}{2}=(\tr[I])^{1/2}$ to normalize $\normp{O}{2}$.
Now we choose a value of $\theta$ close to 1: \begin{equation}
    n = \frac{1}{1-\theta} := \max\left(1,|\ln(\fnorm{O})| \right)  \le |\ln(\fnorm{O})|+1.  \label{eq:F_3proplast1}
\end{equation}
(Note that $\fnorm{O}\le 1$ since $\lVert O\rVert_\infty \le 1$.) Then \begin{equation}
    \fnorm{HO} \le \lVert H\rVert_{(2)}\fnorm{O} \times n \fnorm{O}^{-1/n} \le 2 \lVert H\rVert_{(2)}\fnorm{O} \times n \exp(|\ln\fnorm{O}|/n). \label{eq:F_3proplast2}
\end{equation}
Combining (\ref{eq:F_3proplast1}) and (\ref{eq:F_3proplast2}), we obtain the advertised result.
\end{proof}

In Appendix \ref{app:klocal}, we extend Lemma \ref{lem:submulti} to both $p$-norms and to $k$-local Hamiltonians.  The price of each of these extensions is a slightly worse prefactor in (\ref{eq:3prop}).  However, these prefactors are mild enough that we can still prove Theorem \ref{thm:pnorm_LC}: a partial extension of the Frobenius light cone to any $p$-norm.  

In Appendix \ref{app:othernorm}, we extend our results to an additional family of physical norms, inspired by \cite{Chen:2020bmq,yin2021finite}, which allow us to straightforwardly extend the results of Appendix \ref{app:klocal} to certain norms which induce physical correlation functions in thermodynamic ensembles. 

These two additional sets of results demonstrate the flexibility of our methods.

\section{Light cones}\label{sec:mainthm}
Now we turn to our main application of the mathematics developed above: the Frobenius light cone for a one-dimensional system with power-law interactions.

\subsection{Frobenius light cone}

For simplicity, suppose that we study a spin chain with a qubit on each site.  Let us define a set of projectors $\mathbb{Q}_i$ for $i=0,1,\ldots,R$ (here $R$ is an integer capturing the number of sites on the chain we wish to study).  On Pauli strings of the form $X^{a_0}_0 X^{a_1}_1 \cdots$ with $a_0=0$ denoting $I$ and $a_0=1,2,3$ denoting $X_0,Y_0,Z_0$, we define: \begin{equation} \label{eq:defQx}
    \mathbb{Q}_x | \lbrace X^{a_i}_i\rbrace ) = \left\lbrace \begin{array}{ll} \indicator(a_y=0 \text{ if } y>0)| \lbrace X^{a_i}_i\rbrace ) &\ x=0 \\  \indicator(a_x >0)\indicator(a_y =0\text{ if } x<y) | \lbrace X^{a_i}_i\rbrace ) &\ 0<x<R \\ \indicator(a_y\ne0 \text{ for some } y\ge R)| \lbrace X^{a_i}_i\rbrace ) &\ x=R \end{array}\right..
\end{equation}
Here $\indicator$ denotes the indicator function which returns 1 if its argument is true, and 0 if false.   The projector $\mathbb{Q}_x$ hence projects onto operators which act non-trivially on site $x$, but trivially on all sites to the right of $x$.  The exceptions are $\mathbb{Q}_0$ and $\mathbb{Q}_R$, where we simply include all operators supported on a site $\ge R$ into $\mathbb{Q}_R$, and all operators supported on $\le 0$ in $\mathbb{Q}_0$.  Note that \begin{equation}
    \sum_{i=0}^R \mathbb{Q}_i = 1.
\end{equation}   
For more general qudit systems, the extension of these definitions is straightforward (albeit notation can get clunkier).

With these projectors defined, we may now state our main result, which is the Frobenius light cone for spin chains with power-law interactions:
\begin{thm}\label{thm4}
Let $H(t)$ be a time-dependent Hamiltonian which has power-law exponent $\alpha>1$, as defined in Section \ref{sec:preLR}; let \begin{equation}
    R = 2^{q_*+1}-1 \label{eq:Rdef}
\end{equation} for some positive integer $q_*$. Then for any $\delta>0$, if $\mathbb{Q}_0|A_0) = |A_0)$,  and
\begin{equation}
    \lVert \mathbb{Q}_R |A_0(t))\rVert_{\mathrm{F}} \ge \delta, \label{eq:delta41}
\end{equation}
then there exists a $\alpha$-dependent constant $0<K_\alpha<\infty$ such that \begin{equation}
    |t| \ge \delta^2 K_\alpha \times \left\lbrace \begin{array}{ll} R/\ln R &\ \alpha>2  \\ R/\ln^2 R &\ \alpha =  2 \\ R^{\alpha-1} &\ 1<\alpha<2 \end{array} \right.. \label{eq:main41}
\end{equation}
\end{thm}
Note that in the statement of this theorem, we choose $R$ to be (nearly) a power of 2 for the conceptual simplicity of Figure \ref{fig:long_range_int}. From the definition in (\ref{eq:defQx}), we can always just choose $r/2 < R \le r$ if we wish to get a bound on the light cone at any distance $r$ which is not of the form (\ref{eq:Rdef}).

\begin{proof}
The method of proof is adapted from the ``quantum walk" proofs developed in \cite{hierarachy,Lucas:2019cxr,Lucas:2020pgj,Yin:2020oze,Yin:2020pjd,Chen:2020bmq}.  In a nutshell, we can think of $(A_0(t)|\mathbb{Q}_x|A_0(t))$ as the probability that the operator has made it as far right as site $x$.  We then try to bound how quickly $(A_0(t)|\mathbb{Q}_x|A_0(t))$ can grow with $x$ and $t$.  We do this using Markov's inequality, by bounding the \emph{expected right-most site} of the operator, and noting that \begin{equation}
    (A_0(t)|\mathbb{Q}_R|A_0(t)) \le \frac{\mathbb{E}[\text{right-most site $x$ at time $t$}]}{R}.
\end{equation}

To formalize this last equation, we define the following superoperator $\mathcal{F}$ which acts on an operator and returns its right most site: \begin{equation}
    \mathcal{F} := \sum_{i=0}^R i \mathbb{Q}_i.
\end{equation} Our goal is to prove that at sufficiently short $t$,
\begin{equation}
    \delta^2 \le \lVert \mathbb{Q}_R |A_0(t)) \rVert_{\mathrm{F}}^2 \le \frac{1}{R}(A_0(t)|\mathcal{F}|A_0(t)).\label{eq:markovdelta}
\end{equation}
We thus must bound $\mathrm{d}/\mathrm{d}t (A|\mathcal{F}|A)$. This is achieved by the following Lemma: \begin{lem} \label{lem42}
For constants $0<C<\infty$,
 \begin{align} \label{eq:lem42eq}
    \frac{\mathrm{d}}{\mathrm{d}t} &(A_0(t)|\mathcal{F}|A_0(t)) 
    \le  \notag \\
    & \left\lbrace \begin{array}{ll} 
    C(q_*+1) \left(4+ (1+ \frac{q_*}{2})\ln 2\right)  &\ \alpha =2 \\
    \displaystyle  \dfrac{C}{(1-2^{2-\alpha})^2}\left((4+\ln 2)(2^{2-\alpha}-1)(2^{(2-\alpha)q_*}-1) + \ln 2 (2^{(2-\alpha)q_*} + q_* - (1+q_*)2^{2-\alpha})\right)&\  \text{else} \end{array}\right..
\end{align} 
\end{lem}
\begin{proof}
The rigorous proof is in Appendix \ref{app:41}; here we sketch the main ideas.    We divide up the Hamiltonian into a set of scales $q$, with the rough intuition that couplings on scale $q$ have lengths of order $2^q$.  Using some identities special to one-dimensional models, we are able to re-write \begin{align}
    \frac{\mathrm{d}}{\mathrm{d}t} (A_0(t)|\mathcal{F}|A_0(t)) &= (A_0(t)|[\mathcal{F},\mathcal{L}]|A_0(t)) = \sum_{q=0}^{q_*} \sum_k (A_0(t)|[\mathcal{F},\mathcal{L}_{q,k}]|A_0(t))
\end{align}
as a sum of similar inner products at each scale $q$; here $q_*$ denotes the maximum scale $q_*\sim \ln r$, and $k$ denotes which ``block" of couplings of size $\sim 2^q$ we are studying.   These blocks are shown in Figure \ref{fig:long_range_int}.  At each scale, we can efficiently re-sum up the magnitude of each $\mathcal{L}_{q,k}$ using Lemma \ref{lem:submulti}.  We find that each factor of $\mathcal{L}_{q,k}$ contributes $(2^q)^{1-\alpha}$, following the logic of (\ref{eq:ball2}).  When $\alpha\ge 2$, the logarithmic factors in (\ref{eq:3prop}) cannot be neglected, and lead to an additional logarithmic enhancement in (\ref{eq:lem42eq}).
\end{proof} 

\begin{figure}[t]
    \centering
    \includegraphics[width=0.9\textwidth]{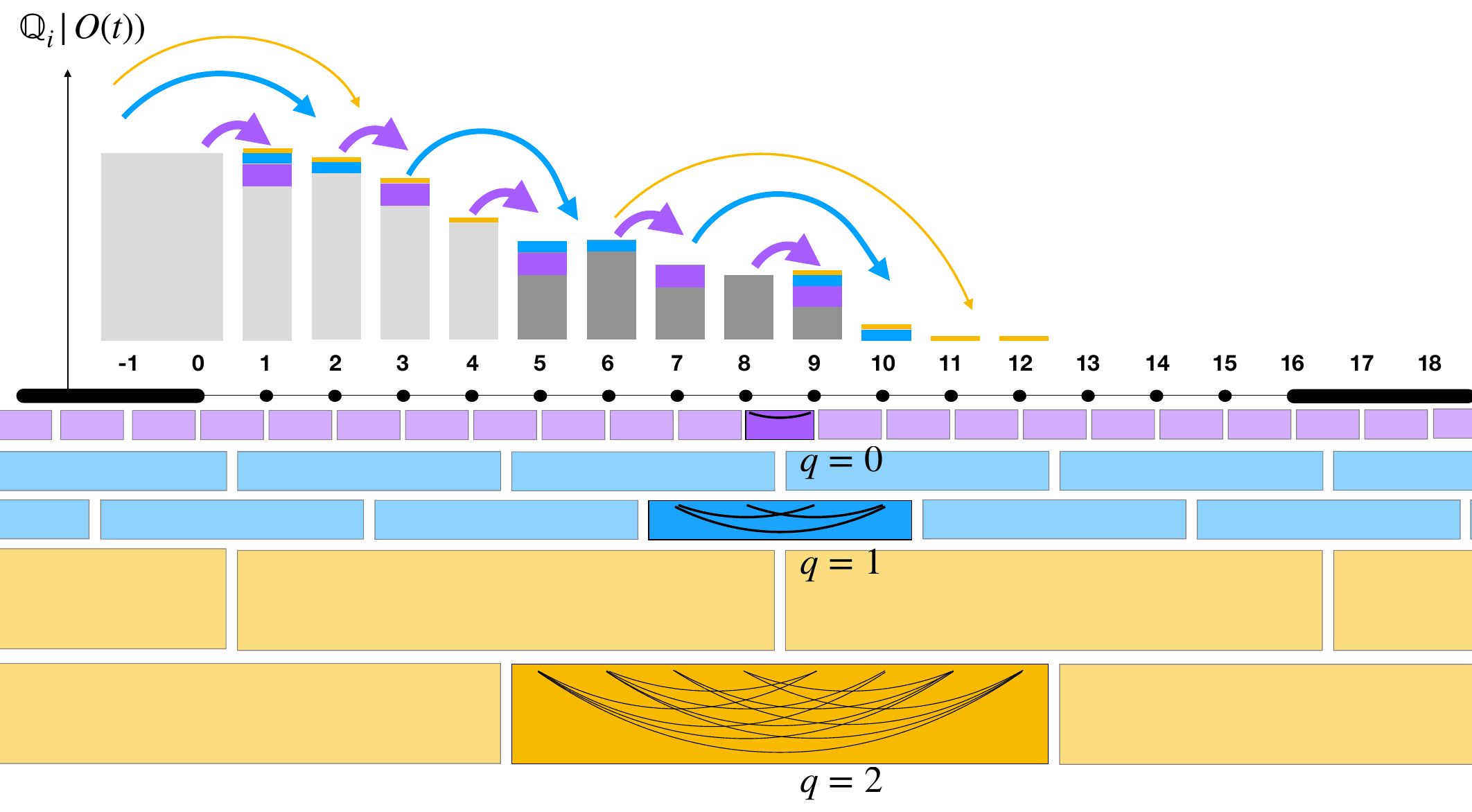}
    \caption{We organize a growing operator $|A_0(t))$ according to the site $x$ on which the right-most non-trivial Pauli acts (achieved via projectors $\mathbb{Q}_x|A_0(t))$. We then regroup the Hamiltonian into scales as in~\cite{alpha_3_chenlucas}.  We then study the quantum walk of the weights $\mathbb{Q}_x|A_0(t))$, which are efficiently bounded one scale at a time.}
    \label{fig:long_range_int}
\end{figure}

Using the result of this lemma, we see that there exists an $\alpha$-dependent constant $K_\alpha$ such that \begin{equation}
    \frac{\mathrm{d}}{\mathrm{d}t} (A_0(t)|\mathcal{F}|A_0(t)) \le \left\lbrace \begin{array}{ll} K_\alpha \ln R &\ \alpha > 2 \\ K_2 \ln^2 R &\ \alpha = 2 \\ K_\alpha R^{2-\alpha} &\ 1<\alpha<2 \end{array}\right.. \label{eq:dfdtbounds}
\end{equation}
Indeed, looking at (\ref{eq:lem42eq}), when $\alpha=2$, the constant scales as $q_*^2 \sim \ln^2 R$; when $\alpha>2$, the term linear in $q_*$ dominates and leads to $\ln R$ scaling; when $\alpha<2$, the $2^{(2-\alpha)q_*}\sim R^{2-\alpha}$ terms dominate.

Combining (\ref{eq:dfdtbounds}) and (\ref{eq:markovdelta}) we obtain (\ref{eq:main41}).
\end{proof}

As we will show in the next section, the scaling of Theorem \ref{thm4} is optimal.  We thus arrive at the surprising result that it is possible to have a Lieb-Robinson velocity which grows exponentially with time (at $\alpha=2$), while the butterfly velocity of the Frobenius light cone grows no faster than $\ln^2 t$.

We note that it is quite challenging to extend this result to higher dimensions.  The reason is that on a 1d chain, it is very simple to define the projectors $\mathbb{Q}_x$ to classify how far to the right an operator has grown.  For $d>1$, we have tried to define analogous projectors keeping track of the ``diameter" of a growing operator, yet because there are $L^{d-1}$ sites near the edge of a ball of diameter $L$, the transition rates between sectors $\mathbb{Q}_L$ and $\mathbb{Q}_{L+1}$ will grow with $L$.  The more sophisticated ansatz necessary to use quantum walk bounds to constrain higher dimensional models (with power-law interactions) has not yet been found.  An analogous challenge was present in the original proof of a linear light cone in 1d systems with power law interactions \cite{alpha_3_chenlucas}; a much more complicated proof \cite{kuwaharaStrictlyLinearLight2020} was necessary to prove the linear light cone in higher $d$.  We would not be surprised if a similar lengthy extension is required to generalize Theorem \ref{thm4} to $d>1$; however, we expect that our Lemma \ref{lem:submulti} will be at the heart of any generalization of this work.

\subsection{The Lieb-Robinson light cone is rarely tight}

A natural question to ask is whether the Frobenius light cone or the Lieb-Robinson light cone is more relevant for a ``typical'' physical state.  Naively, one would expect the Frobenius light cone to be more relevant. Indeed, we can take a probabilisitc interpretation of the Frobenius norm:
\begin{align}
    \fnorm{A}^2=\tr[ \BE_{\psi}(\ket{\psi}\bra{\psi}) A^\dagger A  ] = \BE_{\psi}[\normp{A\ket{\psi}}{\ell_2}^2 ]
\end{align}
for any ensemble of pure states that average to the maximally mixed $\BE_{\psi}(\ket{\psi}\bra{\psi}) = I/\tr[I]$.  
Then, we can obtain a concentration inequality via Chebychev's inequality:
\begin{align}
        \mathbb{P}\L(  \lnormp{ A \ket{\psi}}{\ell_2}\ge a \R) \le \frac{\fnorm{A}^2}{a^2}.
\end{align}
In words, a state drawn ``randomly'' from the ensemble will mostly like be of order $\CO(\fnorm{A})$. However, $1/a^2$ dependence is not the strongest concentration we can ask for.

As it turns out, a sharper concentration inequality would be possible if we further obtained bounds on the Schatten $p$-norms of $A$, for tunable values of $p$.  Let us denote normalized p-norm by \begin{equation}
  \normp{O}{\bar{p}}:=\frac{\normp{O}{p}}{\normp{I}{p}}  \label{eq:barnorm}
\end{equation} that avoids tedious and exponentially large normalization constants when discussing the $p$-norms of local operators in a many-body system.  We now note the following fact:
\begin{prop}[p-norm and typical states~\cite{chen2021trotter}]\label{fact:pnorm_typical_state} For any operator $A$, when uniformly averaging over states $|\psi\rangle$ that $\BE \ket{\psi}\bra{\psi} = I/\tr[I]$,
\begin{align}
    \mathbb{P}\L(  \lnormp{ A \ket{\psi}}{\ell_2}\ge a \R) \le \L(\frac{\normp{A}{\bar{p}}}{a}\R)^p.
\end{align}
\end{prop}

And we obtain a p-norm estimate as follows.
\begin{thm} \label{thm:pnorm_LC}
Let $H(t)$ be a time-dependent Hamiltonian which has power-law exponent $\alpha>1$, as defined in Section \ref{sec:preLR}.  Then under the Heisenberg equation of motion generated by this $H(t)$, 
\begin{align}
      \lVert \mathbb{Q}_R A_0(t) \rVert_p \le c' \sqrt{p}\frac{\labs{t}}{\mathcal{R}(r)}
\end{align}for constant $0<c^\prime(\alpha)<\infty$ which only depends on $\alpha$, and
\begin{align} \label{eq:mathcalR}
    \mathcal{R}(r) = \begin{cases}
     r &\text{if } \alpha> 5/2\\
     r/\ln^{3/2}(r) &\text{if } \alpha= 5/2\\
     r^{\alpha-3/2} &\text{if } 5/2>\alpha> 3/2
    \end{cases}.
\end{align}
Alternatively, we may write that 
$
    \lVert \mathbb{Q}_R A_0(t) \rVert_p \ge \delta 
$
is only possible if \begin{equation}
    |t| \ge \delta\sqrt{p}  c' \mathcal{R}(r). \label{eq:main43}
\end{equation}
\end{thm}
\begin{proof}
The technical details of the proof are quite similar in spirit to the proof of \cite{alpha_3_chenlucas}, which itself uses the same decomposition of $H$ into different scales shown in Figure \ref{fig:long_range_int}.  The difference between the proof of this theorem and Theorem \ref{thm4} is that we are unable to use the quantum walk bounds to tightly control the growth of $ \mathbb{Q}_R |A_0(t))$, because $\lVert  \mathbb{Q}_x |A_0(t)) \rVert_{\bar p}$ does not obey a quantum walk equation.  We give the proof in Appendix \ref{app:proof_pnorm}.  
\end{proof}
We do not expect this bound to be tight -- clearly it is not for $p=2$, in which case Theorem \ref{thm4} is already stronger. Nevertheless, it already suffices to guarantee a meaningful concentration bound:

\begin{cor}\label{cor:concentration_LC}
For power-law models with exponent $2<\alpha<3$, the Lieb-Robinson light cone is rarely saturated: for $0<\epsilon<\infty$ and sufficiently large $r$,
\begin{align}
    \mathbb{P}\L(  \lnormp{ [A_0(t),B_r] \ket{\psi}}{\ell_2}\ge \frac{\epsilon t}{r^{\alpha-2}}\R) \le  \exp(2-\epsilon^2 Cr^{\beta}) \end{align}
for some constant $0<C(\alpha)<\infty$ which does not depend on $r$, and with $\beta = \min(1,6-2\alpha)-\delta$ for arbitrary $\delta>0$. This result holds even if $H(t)$ is time-dependent.
\end{cor}
Note that the infinitesimal parameter $\delta$ is only due to the logarithm $\ln(r)$ at $\alpha=5/2$, and should not distract the reader.
\begin{proof}
Plug the $p$-norm bound of Theorem~\ref{thm:pnorm_LC} into Proposition~\ref{fact:pnorm_typical_state} to obtain 
\begin{align}
    \mathbb{P}\L(  \lnormp{ [A_0(t),B_r] \ket{\psi}}{\ell_2}\ge \frac{\epsilon t}{r^{\alpha-2}} \R) &\le \left(\frac{r^{\alpha-2}}{\epsilon t} \cdot \sqrt{p}\cdot \frac{c't}{\mathcal{R}(r)}\right)^p = \left(\sqrt{p}\cdot \frac{c'r^{\alpha-2}}{\epsilon \mathcal{R}(r)}\right)^p
    \notag \\
    &\le \exp\left(2 -\epsilon^2\frac{\mathcal{R}(r)^2r^{4-2\alpha}}{\e^2c'^2}\right).
\end{align}
where in the last inequality we used that $(\sqrt{p} b)^p \le \exp(2-\mathrm{e}^{-2}b^{-2})$ (set $p=\min(2,(\mathrm{e}b)^{-2})$). 
The $+2$ in the exponential is added since to use our concentration bounds we must always take $p\ge 2$. We conclude the proof by noting that for any $2<\alpha<3$,  $r^{4-2\alpha}\mathcal{R}(r)^2>r^\beta$ if $r$ is sufficiently large.
\end{proof}

Remarkably, even though the form of the bound in Theorem \ref{thm:pnorm_LC} is not tight, we find an exponential bound on the number of states in which the Lieb-Robinson bounds might be saturated.  In the regime $2<\alpha<\frac{5}{2}$, this bound is optimal(up to constants in the exponent): the protocol of \cite{2020arXiv201002930T} shows that in a spin chain of qubits, it is possible to saturate the Lieb-Robinson light cone in at least one state.   What Corollary \ref{cor:concentration_LC} proves is that the state identified in \cite{2020arXiv201002930T} is in fact one of the \textit{exponentially} rare states for which the Lieb-Robinson light cone can be saturated!  This strongly suggests that the rapid single-bit state transfer protocol of \cite{2020arXiv201002930T}, which can transfer one qubit in time $t\sim r^{\alpha-2d}$, may require exquisite control of the background state;
this contrasts with the more robust (yet slower) protocol of \cite{hierarachy,yifan}, which uses long-range interactions to send qubits in a ``self-error-correcting" scheme.   It is an interesting open question to more quantitatively compare fast quantum error correction schemes to the optimal Frobenius and Lieb-Robinson light cones which have been developed over the past few years.

Lastly, we conjecture that the optimal bound on $p$-norms takes the form \begin{equation}
      \lVert [A_0(t),B_r] \rVert_{\bar p} \stackrel{?}{\lesssim} \frac{pt}{r^{\alpha-1}}.
\end{equation}
This is an educated guess: it extrapolates to the $p=2$ case, and that the concentration holds at $\e^{-\Omega(r)}$ though Corollary~\ref{cor:concentration_LC}. We do not currently know how and if our quantum walk based proof of Theorem \ref{thm4} can be generalized to a $p$-norm bound.
 
\section{Algebraically optimal operator growth protocol}\label{sec:protocol}
Having established the Frobenius light cone rigorously in $d=1$, and with a conjecture on how it generalizes to higher dimensions $d$, let us now describe a protocol which we claim (and will susbequently prove) achieves these speed limits (up to sub-algebraic corrections).   Our approach is loosely inspired by the optimally fast single-qubit state transfer protocol developed in \cite{2020arXiv201002930T}:  as in \cite{2020arXiv201002930T}, we will develop our protocol via ``recursive" intuition.

\subsection{Intuitive argument}
To begin, let us assume that we have a system with tunable and time-dependent power-law interactions of exponent $d<\alpha<d+1$ on the standard hypercubic lattice $\mathbb{Z}^d$.   We divide up this lattice into a partition of hypercubes at multiple scales $q=0,1,2,\ldots$.  For intuitive purposes, we can say that the scale $q=0$ corresponds to each lattice site being in its own cube; scale $q=1$ corresponds to a partition of the lattice into hypercubes of side length $m_1$ in all dimensions; scale $q=2$ corresponds to a partition into hypercubes of side length $m_1m_2$ in all dimensions, and so on.  (Note that the $q=2$ cubes contain $m_2$ $q=1$ cubes within them, etc.).   In the discussion that follows, we will for simplicity set $m_1=m_2=\cdots = m$.

At time $t=0$, we start with a Pauli matrix $X_0$ on a single site.  By the definitions above, that Pauli matrix occupies exactly one 0-cube, which we might as well call the origin of the lattice.  Our goal is to find a quantum mechanical protocol (i.e. a unitary matrix $U(t)$, which can be generated from a power-law Hamiltonian evolving for time $t$, possibly with time-dependent coefficients), such that $U(t)^\dagger X_0U(t)$ consists of Pauli strings of length $L^d$.  If we achieve this, then we know that this operator must have support on at least one site a distance $\sim L$ away from the origin where we started.  Our goal is to do this in a time $t\sim L^{\alpha-d}$ for $d<\alpha<d+1$.  

To motivate how we might achieve this task, suppose that we have a quantum protocol -- a unitary matrix $U_q$ -- which is capable of taking any single-site Pauli matrix (e.g. $X_0$) and evolving it into an operator $U_q^\dagger W U_q$ which is supported on an $\CO(1)$ fraction of sites in the $q$-cube $C_q$: \begin{equation}
    \sum_{x\in C_q} \left\lVert \mathbb{P}_x U_q^\dagger W U_q \right\rVert^2_{\mathrm{F}} \propto R_q^d, \label{eq:intuit}
\end{equation}
with $R_q^d$ being the number of sites in $C_q$, and \begin{equation}
    R_q= m R_{q-1}= m^q . \label{eq:Rqdef}
\end{equation}  
We can certainly do this in constant time when $q=1$ by simply using nearest neighbor interactions.  

Now, assuming that we found $U_q$, let us find a $U_{q+1}$ at the next scale. First, note that any $q$-cube $C_q$ lies entirely within a $(q+1)$-cube $C_{q+1}$.  How can we find a unitary $U_{q+1}$ such that (\ref{eq:intuit}) continues to hold at scale $q+1$?   One possible way to do this would be: \begin{equation}
    U_{q+1} = U_q V_{q+1}U_q,
\end{equation}
where $V_{q+1}$ is a unitary that takes an operator supported in a single $q$-cube, and evolves it to have support on only a single site in each of the other $m^d-1$ $q$-cubes in $C_{q+1}$.   The three-step process is sketched in Fig.~\ref{fig:intuitive_protocol}.
\begin{figure}[t]
    \centering
    \includegraphics[width=0.9\textwidth]{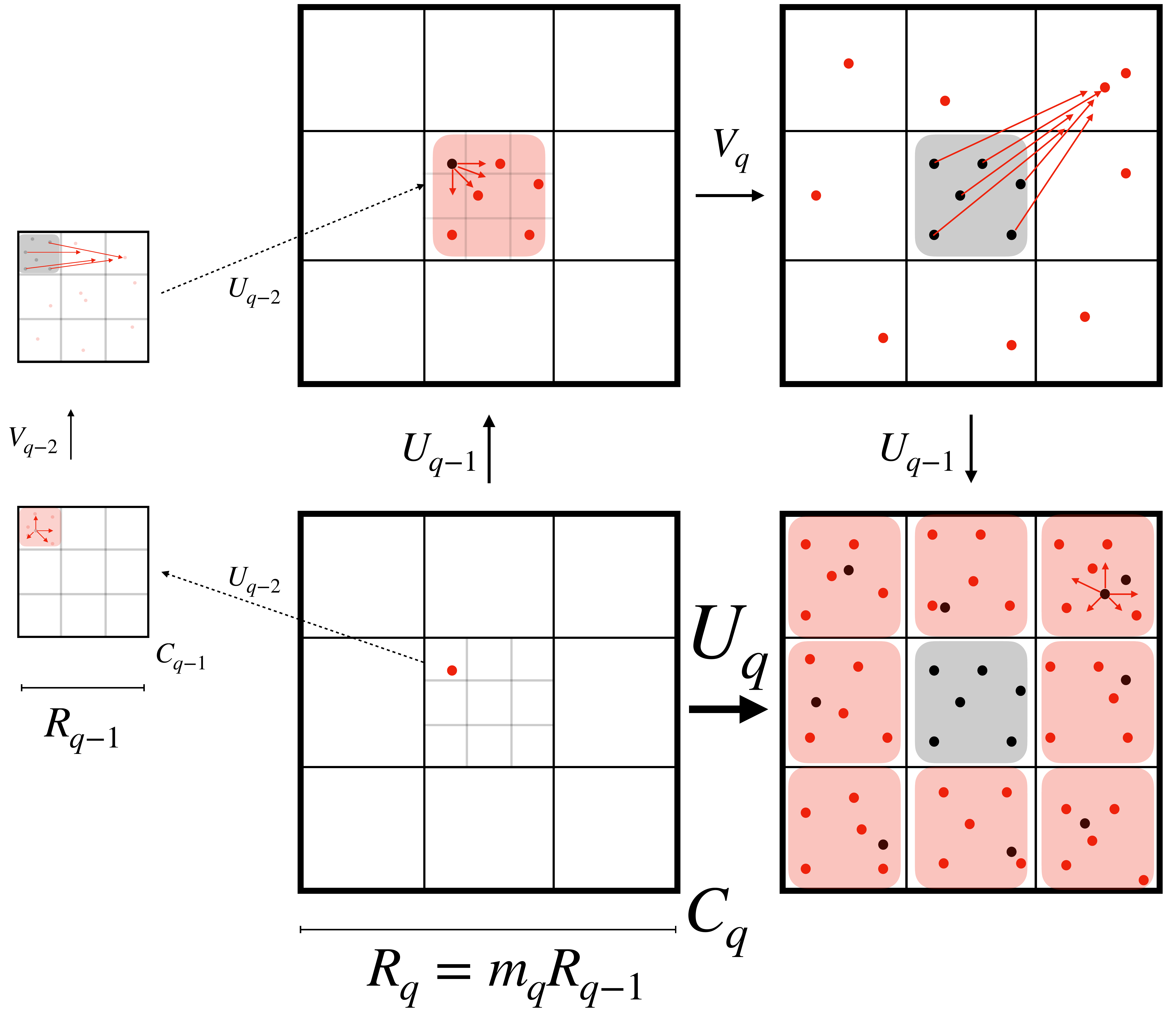}
    \caption[f]{The recursive step in the idealized protocol, constructing $U_q$ from $U_{q-1}$. The new sites are shaded lighter. (1)  grow a single Pauli $W$ to have support on $\sim |C_q|$ sites using the first $U_{q-1}$. (2) ``seed" \textit{a single} Pauli matrix (e.g. $Z_x$) on one site $x$ in each other $q$-cube $C^\prime_q\subset C_{q}$ using $V_{q}$. (3) run again the $U_{q-1}$ protocol on each remaining $q$-cube $C_q^\prime$ to ``bloom" the seeded single operator into a finite fraction of sites.
On the left is the description of $U_{q-1}$. 
    }
    \label{fig:intuitive_protocol}
\end{figure}

The key observation is that this ``recursive" construction is quite natural to implement with power-law interactions.  As the scale $q=0\rightarrow 1$ unitary $U_1$ is quite easy to implement (even using only nearest neighbor interactions), let us focus on how $V_{q+1}$ might effectively be implemented with power-law interactions.   Suppose that as an idealistic cartoon, we found that \begin{equation}
    U_q^\dagger X_0 U_q = \prod_{x\in C_q} X_x.
\end{equation}
Now let us consider \begin{equation}
    V_{q+1} = \exp\left[-\mathrm{i}\tau_{q+1} \sum_{x,y\in C_{q+1} } \frac{Z_x Z_y}{(dR_{q+1})^\alpha} \right]. \label{eq:Vintuit}
\end{equation}
Clearly, this Hamiltonian is compatible with (\ref{eq:Jab}).   When $\tau_{q+1}$ is small, we may estimate that the fraction of the operator with support outside of the original $C_q$ is given by the first order expansion \begin{align} \label{eq:56}
    \left\lVert \tau_{q+1} \left[\sum_{x,y\in C_{q+1} } \frac{Z_x Z_y}{(dR_{q+1})^\alpha}, \prod_{x\in C_q} X_x \right]\right\rVert_{\mathrm{F}}^2 &=  \left\lVert \sum_{x\in C_q, y\in C_{q+1}-C_q} \frac{2\tau_{q+1}Y_xZ_y}{(dR_{q+1})^\alpha} \prod_{z\in C_q-x} X_z \right\rVert_{\mathrm{F}}^2 \notag \\
    &= \frac{4\tau_{q+1}^2}{(dR_{q+1})^\alpha}R_q^d \left(R_{q+1}^d - R_q^d\right)
\end{align}
Taking \begin{equation}
    \tau_{q+1}\gtrsim R_{q+1}^{\alpha-d},
\end{equation}
we can estimate that our protocol may have substantial weight outside of $C_q$(and thus have some seeds in each new $C_{q}$.).  The runtime of the overall $U_{q+1}$ is, recursively, \begin{equation}
    t_{q+1} = \tau_{q+1}+ 2 t_q. \label{eq:tq1intuit}
\end{equation}
And indeed, \begin{equation}
    t_q \propto R_q^{\alpha-d} 
\end{equation}
does appear consistent with (\ref{eq:tq1intuit}).   This suggests that our recursive approach will be capable of growing an operator at the provably optimal rate in $d=1$, and conjectured optimal rate in $d>1$.

In order to make the argument above precise, we will study a random ensemble of protocols, which proceed up to scale $q\sim \ln m\sim \sqrt{\ln r}$.  This peculiar choice of $m$ happens to speed up some technical steps in our proof, while only leading to sub-algebraic factors in the overall runtime of the method.  We will show that \emph{at least half of the protocols in our random ensemble grow operators quickly}, (nearly) saturating Theorem \ref{thm4}.  We emphasize that randomness does not play some key physical role -- after all, if we just pick some random instance of our protocol, with at least 50\% chance that one non-random protocol can grow operators quickly!  Randomness is useful, however, as a way to simplify what is already a rather tedious proof which aims to keep track of the probability that a growing operator $O(t)$ has support on a given lattice site.   In a nutshell, we introduce random depolarizing unitaries (denoted as $D$ below) which will allow us to keep track of only whether a given site has a Pauli matrix on it or not, but ignore what kind of Pauli we have.   This makes our technical proof a bit more conceptually simple.  

Despite this randomness, we emphasize that constructive quantum interference is absolutely essential to the success of these random protocols.  Indeed, (\ref{eq:56}) holds as $\tau_{q+1}\rightarrow 0$; however, we also need to make sure that (\ref{eq:56}) holds at sufficiently late times such that the unitary $V_{q+1}$ has spread the operator $U_q^\dagger X_0 U_q$ into many sites in $C_{q+1}$.  The fact that $V_{q+1}$ is the product of many $\exp[-\mathrm{i}ZZ]$ unitaries, which are all mutually commuting, ensures that the long-range interactions are interfering \textit{coherently}.

\subsection{Explicit protocol}\label{sec:explicit}
Having gone through the intuitive argument, let us now present an explicit random Hamiltonian protocol that achieves this result -- at the expense of a logarithmically increased runtime (which we will ``justify" at the end of this section).   The protocol is built using the same $q$-cube structure outlined above.   In a nutshell, our $(q+1)$-scale protocol takes the form \begin{equation}
    U_{q+1} = U_q^\prime V_{q+1} U_q, \label{eq:induction}
\end{equation}
where $U_q$ and $U_q^\prime$ are \textit{random} $q$-scale unitary protocols drawn from an ensemble that we will state below, and \begin{equation}
    U'_{q} = U_{q}D,
\end{equation}
where $D$ is a depolarizing unitary drawn uniformly at random from a discrete ensemble (to be described below) and $V_{q+1}$ is a random unitary, built out of power-law interactions, that will mimic (\ref{eq:Vintuit}). At each step, we will choose the unitaries $U_q$, $D$, etc. uniformly at random from particular distributions.  Perhaps surprisingly, this randomness will actually \emph{help} us prove that the protocol works. The depolarizer $D$ will allow us to effectively ignore what operator is present on a given site, and just keep track of which sites our growing operator has support on.  Analogous to recent work on random unitary circuits \cite{Nahum:2017yvy,vonKeyserlingk:2017dyr}, this can make it much easier to keep track of the many-body operator growth. 

Ultimately, we will prove that the ensemble averaged Frobenius norm of a commutator is large; if the average is large, there must be one instance of a  unitary time evolution operator in the ensemble (which we do not need to explicitly point out), which achieves a large Frobenius norm.  Of course, it is even more interesting that \emph{typical} unitaries in the ensemble achieve a large Frobenius norm.

Let us now carefully define the $q$-cube partitions of the lattice $\mathbb{Z}^d$.   We define the $q$-cubes \begin{equation}
    C_q(k_1,\ldots,k_d) := \lbrace (n_1,\ldots, n_d) \in \mathbb{Z}^d : R_q k_i \le n_i < R_q (k_i+1) \rbrace.
\end{equation}
The set of all such $q$-cubes will be denoted with \begin{equation}
    \mathcal{B}_q := \lbrace C_q(\mathbf{k}) : \mathbf{k}\in\mathbb{Z}^d\rbrace.
\end{equation}
$R_q$ is defined via (\ref{eq:Rqdef}), where \begin{equation}
    m:= \left\lceil \mathrm{e}^{\sqrt{\ln(r)}}\right\rceil. \label{eq:mdef}
\end{equation}

The random depolarizing unitary $D$ is chosen as follows: \begin{equation}
    D := \bigotimes_{x\in \Lambda} D_x, \label{eq:Dx}
\end{equation}where $D_x$ are $2\times 2$ unitary matrix acting on qubit $x$, independent and identically distributed (iid) for each $x$.  Each $D_x$ is chosen uniformly at random (via the discrete Haar measure) from the group \begin{equation}
    G = \left\lbrace \pm 1, \frac{\pm 1 \pm \mathrm{i}X^a}{\sqrt{2}}, \pm \mathrm{i}X^a, \frac{\pm \mathrm{i}X^a \pm \mathrm{i}X^b}{\sqrt{2}}, \frac{\pm 1 \pm \mathrm{i}X \pm \mathrm{i}Y \pm \mathrm{i}Z}{2}\right\rbrace \label{eq:groupG}
\end{equation} 
In the above equation, $X^a$ and $X^b$ denote \emph{distinct} Pauli matrices ($X,Y,Z$).   Each $D_x$ can be generated using a single-site Hamiltonian of (operator) norm 1 in a time \begin{equation}
    t_D < \frac{\pi}{2}.
\end{equation}

The growth unitary is \begin{equation}
    V_{q} = \exp[-\mathrm{i}\tau_{q+1} H^{ZZ}_{q}],
\end{equation}
where the Hamiltonian \begin{equation}
    H^{ZZ}_{q} := \frac{1}{(dR_q)^\alpha} \sum_{C \in \mathcal{B}_q} \sum_{x,y \in C} J_{xy}Z_xZ_y, \label{eq:HZZq}
\end{equation}
where $J_{xy}$ are iid random variables uniformly distributed on the interval $[-1,1]$.  This clearly mimics what we intuitively introduced above; however, we will see that the randomness in the couplings is \emph{beneficial} in allowing us to neglect possible quantum interference phenomena (among growing operators) that might ruin our protocol.  The times $\tau_q$ will be chosen explicitly in (\ref{eq:tauqdef}) in Appendix \ref{app:5}, but note for now that it obeys \begin{equation}
    \tau_q < 120^q m^d R_q^{\alpha-d} . \label{eq:tauqbound}
\end{equation}

The protocol will stop at scale \begin{equation}
    q_* := \left\lceil \sqrt{\ln r} \right\rceil. \label{eq:qstardef}
\end{equation}
At this scale, at least half of the operator (as measured by the Frobenius norm) will have size $>r^{d-\epsilon}$ for any $\epsilon$ (Proposition \ref{thm5}).   In order to demonstrate that the Frobenius light cone is saturated (up to subalgebraic prefactors), we must calculate the total runtime of the protocol $U_q$.  Using the inductive identity (\ref{eq:induction}), along with (\ref{eq:tauqbound}), we see that if $t_q$ is the total runtime of $U_q$, and $t_D < \tau_q$ is the runtime of $D$,\footnote{This condition will always be assured at sufficiently large $r$.  We can shrink the prefactor of $H^{ZZ}_q$ to ensure this at small $r$.} then
\begin{align}
    t_q  &= 2t_{q-1} + t_D + \tau_q < 2t_{q-1}+2\tau_q = \sum_{q=1}^{q_*} \tau_{q_*-q}2^{q} < 2^{q_*+1}\tau_{q_*}  \notag \\
    & < 480\cdot (240\mathrm{e})^{\sqrt{\ln r}} \left(1+ \mathrm{e}^{\sqrt{\ln r}}\right)^{(\alpha-d)(1+\sqrt{\ln r})} < 480\cdot (240\mathrm{e}^{\alpha-d+1}\cdot 4^{\alpha-d})^{\sqrt{\ln r}} r^{\alpha-d}.
\end{align}

We thus conclude a \emph{lower bound} on the Frobenius light cone in any dimension.  To be more precise, given the decomposition of an operator defined in (\ref{eq:AS}), let us define the projector \begin{equation}
    \mathbb{P}_{\ge L} |A) := \sum_{S\subset \Lambda : \mathrm{diam}(S\cup \lbrace 0 \rbrace) \ge L} |A_S).
\end{equation}
Here $\mathrm{diam}(S)$ denotes the maximal distance between any two elements in the set $S$.  The discussion above immediately implies the following theorem:
\begin{thm}\label{thm51}
Let $X_0$ be the Pauli $X$-matrix supported at the origin $(0,\ldots,0) \in \mathbb{Z}^d$.  For any $\epsilon>0$, define $r(L) := L^{1+\epsilon/2}$.  Then there exists a sufficiently large $L$, such that a power-law Hamiltonian protocol drawn from the distribution (\ref{eq:induction}) achieves \begin{equation}
    \lVert \mathbb{P}_{\ge L} |U^\dagger_{q_*} X_0 U_{q_*}) \rVert_{\mathrm{F}} \ge \frac{1}{2}, \label{eq:bigpropeqn}
\end{equation}
with probability $\ge 1/2$. Moreover, the shape of the Frobenius light cone is bounded by \begin{align}
    L(t) \ge K_\epsilon t^{\frac{1}{\alpha-d} - \epsilon} 
\end{align}
for some constant $0<K_\epsilon<\infty$. The asymptotic bound of Theorem \ref{thm4} cannot be improved by an algebraic factor.
\end{thm}

The last thing we need to do is to prove that our protocol in fact does grow a finite fraction of an initial small operator to be large.  This result is captured by the following key technical proposition: 

\begin{prop}\label{thm5}
For sufficiently large $L$, there exists a Hamiltonian in the random ensemble of Section \ref{sec:explicit} in which (\ref{eq:bigpropeqn}) holds.
\end{prop}
\begin{proof}
The proof of this proposition, of course, corresponds to the overwhelming majority of the proof of Theorem \ref{thm51}.   As it is rather technically involved, let us outline the key steps in the proof.  (\emph{1}) We will first develop a ``super-operator density matrix" perspective for bounding the Frobenius light cone.   (\emph{2}) This notation will prove highly useful, since we will show that the ``super-depolarizing channel" (conjugation by $D$) destroys all (unwanted) quantum coherence, and leaves us with an effectively classical problem to analyze.  (\emph{3}) We will then describe the inductive hypothesis required to achieve (\ref{eq:bigpropeqn}), and reduce the quantum mechanical problem of bounding operator growth to the bounds on a classical stochastic process.  (\emph{4}) We will then show that (in the language of the effective stochastic process), with very high probability the $V_q$ unitary ``seeds" enough Pauli matrices in new cubes.  (\emph{5}) Analogously, we will show that with very high probability, these seeds in turn grow into large Pauli strings upon applying another $U_{q-1}$.   (\emph{6}) We will show that at every possible step of the protocol, the collective failure probability is small and decreases sufficiently fast that there is a finite success probability to grow a large operator.  Upon converting back to the quantum mechanical language, that will imply (\ref{eq:bigpropeqn}).

The technical implementation of this proof is in Appendix \ref{app:5}.
\end{proof}

 The fact that \emph{typical} protocols actually are effective at growing operators at the Frobenius light cone suggests a certain robustness of this notion of light cone.  At the same time, for chaotic systems, the light cone is linear above $\alpha>d+\frac{1}{2}$ \cite{chen2021concentration, zhou}.  One way to reconcile these facts is to note that the chaotic Hamiltonians do not choose $ZZ$ Hamiltonians to implement operator growth in $V_q$.  By only including $ZZ$ terms in our long-range Hamiltonian, we ensure that there is always constructive interference \emph{in the $V_q$ step} at which the operator actually grows.\footnote{However, our $\mathcal{V}_q$ is incoherent between different interaction terms due to the random coefficients. This does not slow down the protocol because different Paulis always add up incoherently (i.e. as sum of squares) in the Frobenius norm, regardless of their phases.}  In a Brownian circuit \cite{chen2021concentration, zhou}, in contrast, the operator is re-chosen randomly at each step in \emph{time}: this leads to far greater incoherence which effectively doubles the power law of the interactions from $\alpha \rightarrow 2\alpha$ (as one can only incoherently add squares).\footnote{These doubled interactions then obey a standard Lieb-Robinson-like light cone, which is linear when $2\alpha>2d+1$, or $\alpha>d+\frac{1}{2}$.}
 

Remarkably, our protocol \emph{also} saturates both the Frobenius light cone of this paper (proven in $d=1$ and conjectured for $d>1$), and a Frobenius light cone proven for random Hamiltonians in $d=1$ (Theorem 7, \cite{chen2021concentration}). In other words, the Frobenius light cones are essentially the same, whether one fine tunes the Hamiltonian or just draws one randomly from an ensemble. We anticipate these conclusions generalize to higher dimensions, though a formal proof is not known.

\section{Outlook}
We have shown constraints on the dynamics of growing operators, measured by the Frobenius norm. In particular, we have proved that in one dimensional spin chains with long-range interactions, it is possible for the Frobenius light cone to be \emph{exponentially slower} than the Lieb-Robinson light cone ($\alpha =2$). Such a result is based on the key insight that Frobenius norm becomes approximately submultiplicative, proven by combining standard and new functional analysis tools (uniform smoothness) with the quantum walk formalism.   

Moreover, we demonstrated that our Frobenius light cone in one dimension is essentially optimal (up to subalgebraic corrections). Our protocol features the first comprehensive analysis of an explicit random Hamiltonian; in contrast, existing results relied on Brownian Hamiltonian dynamics \cite{zhou}. Our usage of super-density operator and super-channel may find further applications in studying operator growth and Frobenius light cones in other systems.

In the near future, we hope to prove our conjecture that the Frobenius light cone in higher dimensional models with long-range interactions looks (schematically like) $t \gtrsim r^{\alpha-d}$ (for $d<\alpha<d+1$).  Beyond that obvious generalization, we anticipate that our novel methods will find use in a broad variety of other challenging problems, such as bounding fast scrambling and chaos in quantum simulators, including trapped ion crystals and cavity quantum electrodynamics \cite{Yin:2020pjd,zehanli,belyansky}; this work may help to constrain when it is possible (or not) to mimic quantum gravity in an experiment \cite{Chew_2017,Chen_2018,Marino_2019,Lewis_Swan_2019,Alavirad_2019,Bentsen_2019prl}.  A more practical  possible application of the Frobenius light cone may be to constrain the generation of (volume-law) entanglement.  Lastly, we hope to develop a more general toolkit (perhaps based on the quantum walk methods) to control Frobenius light cones in arbitrary many-body models.

\section*{Acknowledgements}
We thank Minh Tran and Alexey Gorshkov for helpful discussions and for collaboration on related work.  AL was partially supported by a Research Fellowship from the Alfred P. Sloan Foundation, and by the Air Force Office of Scientific Research under Grant FA9550-21-1-0195.

\begin{appendix}

\section{Extension of Lemma \ref{lem:submulti} to $p$-norms and $k$-local Hamiltonians}\label{app:klocal}

In this appendix, we show that the approximate form of submultiplicativity derived in Section \ref{sec:submulti} extends to Schatten $p$-norms with $p>2$, using the following version of \textit{uniform smoothness}. We include a minimal review, and e.g.~\cite{HNTR20:Matrix-Product,chen2021trotter} for further discussions. \footnote{This particular form of uniform smoothness was perceived when this work and another work~\cite{chen2021trotter} was developing. We include the same proof at both papers.}
\begin{prop}[Uniform smoothness for subsystem]\label{prop:unif_subsystem2}
Consider matrices $X, Y$ of the same dimensions that satisfy
$\tr_i(Y) = 0; X= X_j\otimes I_i$. For $p \ge 2$, 
\begin{equation}
\lVert X + Y\rVert_{p}^2\le \lVert X \rVert_{p}^2  + (p-1)\lVert Y\rVert_{p}^2 .
\end{equation}
\end{prop}
The proof of this proposition, which adapts from proof in \cite[Prop~4.3]{HNTR20:Matrix-Product}, is delayed slightly.  We first recall the following identity: 
\begin{prop}[Uniform smoothness for Schatten Classes{\cite[Fact~4.1]{HNTR20:Matrix-Product}}]\label{fact:unif_schatten}
\begin{align}
    \left[\frac{1}{2}(\normp{X+Y}{p}^p+\normp{X-Y}{p}^p)\right]^{2/p} \le \normp{X}{p}^2+ C_p \normp{Y}{p}^2
\end{align}
\end{prop}

With these facts in hand, we are now ready to generalize Proposition \ref{prop:unif_subsystem} to the $p$-norm.  We begin with the following simple observation:
\begin{prop}\label{prop:nc_convexity}
For $\tr_i(Y) = 0; X= X_{-i}\otimes I_i$, $p \ge 2$,
\begin{align}
\normp{X}{p}\le \normp{X+Y}{p}
\end{align}
\end{prop}
\begin{proof}
We employ a variational formulation for Schatten p-norms \cite{Wilde_QShannon}:
\begin{align}
    \normp{X_j}{p} = \sup_{\normp{B_j}{q}\le 1} \tr(X_j^\dagger B_j)
\end{align}
for $1/p+1/q=1$. Then we restrict to $B$ which are proportional to the identity on site $i$:  $B\propto B_{-i}\otimes I_i$:
\begin{align}
    \normp{X+Y}{p} = \sup_{\normp{B}{q}\le 1} \tr((X+Y)^\dagger B) \ge \tr\left( (X^\dagger+Y^\dagger) B_{-i}\otimes \frac{I_i}{\normp{I_i}{q}}\right) = \normp{X_{-i}\otimes I}{p}
\end{align}
In the last inequality we used that $\tr_i Y =0$ and that $X_{-i}\otimes I$ has maximizer $B_{-i}\otimes I_i/{\normp{I_i}{q}}$.
\end{proof}
We can now prove Proposition~\ref{prop:unif_subsystem2}.
\begin{proof}Observe that
\begin{align}
     \frac{\normp{X+Y}{p}^2 + \normp{X}{p}^2}{2} &\le \frac{\normp{X+Y}{p}^2 + \normp{X-Y}{p}^2}{2} \\
     &\le
     \left(\frac{\normp{X+Y}{p}^p+\normp{X-Y}{p}^p}{2} \right)^{2/p} 
     \le \normp{X}{p}^2+ C_p \normp{Y}{p}^2.
\end{align}
The last line uses Lyapunov's inequality, and then Proposition~\ref{fact:unif_schatten}. Rearranging terms yields a slightly worse constant $2C_p$. The advertised constant can be obtained via a more involved trick~\cite[Lemma~A.1]{HNTR20:Matrix-Product}.
\end{proof}

Now we can show submultiplicativity for arbitrary $p$-norms.  We will also, for good measure, describe how to generalize our result to $k$-local Hamiltonians as well. As in the main text, we need to properly normalize the Schatten $p$-norms with the ``bar norm" (\ref{eq:barnorm}). 
\begin{prop}\label{prop:submul_pnorm}
For $k$-local Hamiltonian
\begin{align}
    H = \sum_{m\ge i_k>\cdots>i_1\ge 1 }  H_{i_1,\cdots,i_k} 
\end{align}
such that $\mathbb{P}_jH_{i_1,\cdots,i_k} = \indicator(j\in \lbrace i_1,\ldots, i_k\rbrace) H_{i_1,\cdots,i_k}$,
for any operator $O$ obeying $\lV O\rV \le 1$ and any $p\ge 2$,
\begin{align}
    \lV HO\rV_{\bar{p}} \le \e p^{k/2} \cdot \lV H \rV_{(2)} \cdot \lV O\rV_{\bar{p}} \big(|\ln\lV O\rV_{\bar{p}}|+1\big)^{k/2} \label{eq:3prop}
\end{align}
where 
\begin{align}
    \lV H \rV_{(2)}:= \sqrt{\sum_{m\ge i_k>\cdots>i_1\ge 1 }  \normp{  H_{i_1,\cdots,i_k}}{\infty}^2}.
\end{align}
\end{prop}
\begin{proof}
Using Proposition \ref{prop:holder}, \begin{equation}
    \lVert HO\rVert_{p} =  \lVert HO\rVert_p \le \lVert H\rVert_{p/(1-\theta)} \cdot \lVert O\rVert_{p/\theta}. \label{eq:p_holder3_2}
\end{equation}
We now bound each of these norms in turn.
We start with $\normp{H}{q}$ (set$q:=p/(1-\theta)$ in what follows):
\begin{align}
    \left\lVert\sum_{m\ge i_k>\cdots>i_1\ge 1 }  H_{i_1,\cdots,i_k}\right\rVert_{q}^2 &\le  \left\lVert\sum_{m-1\ge i_k >\cdots>i_1\ge 1 }  H_{i_1,\cdots,i_k}\right\rVert_{q}^2 + C_q \left\lVert\sum_{m=i_k> i_{k-1}\cdots>i_1\ge 1 }  H_{i_1,\cdots,i_k}\right\rVert_{q}^2 \notag \\
    &\le \sum_{m \ge i_k \ge k} C_q \left\lVert\sum_{i_k> i_{k-1}\cdots>i_1\ge 1 }  H_{i_1,\cdots,i_k}\right\rVert_{q}^2 \le \sum_{m\ge i_k>\cdots>i_1\ge 1 } {C_q}^k \normp{  H_{i_1,\cdots,i_k}}{q}^2 \\
    &\le {C_q}^k\normp{I}{q}^2\sum_{m\ge i_k>\cdots>i_1\ge 1 }  \normp{  H_{i_1,\cdots,i_k}}{\infty}^2 = {C_q}^k\normp{I}{q}^2 \normp{H}{(2)}^2 \label{eq:p_Hnorm}
\end{align}
In the first line we applied Prop.~\ref{prop:unif_subsystem2} for $Y=\sum_{m=i_k> i_{k-1}\cdots>i_1\ge 1 }  H_{i_1,\cdots,i_k}$ and $X=\sum_{m-1\ge i_k >\cdots>i_1\ge 1 }  H_{i_1,\cdots,i_k}$, i.e. we are ``peeling off" any term traceless on the qudit $i_k=m$. In the second inequality we peel off qudit $i_k=m-1$, qudit $i_k=m-2$ all the way to $i_k=k$.\footnote{The $i_k=k$ term does not have the $C_q$ coefficient, but we threw it in to simplify the expression.} In the second line, we recursively repeat the first and second line for $i_{k-1},\cdots, i_1$, each of which gives a factor of $C_q$. Lastly, we use H\"older's inequality to bound the q-norm by the $\infty$ norm with a $\normp{I}{p}$ overhead.\footnote{This conversion to operator norm is tight when the spectrum is suitably flat, as is the case when dealing with a Pauli string operator.}

For the norm of $O$, a standard manipulation using Proposition \ref{prop:riesz}, and $\lVert O \rVert_\infty \le 1$, allows us to handle \begin{equation}
    \lVert O\rVert_{p/\theta} \le \lVert O\rVert_p^\theta \lVert O\rVert_\infty^{1-\theta}  \le \lVert O\rVert_{p}^\theta  \label{eq:p_Onorm}
\end{equation}
Now combining (\ref{eq:p_holder3_2}), (\ref{eq:p_Hnorm}) and (\ref{eq:p_Onorm}), we plug in $q=p/(1-\theta)$ to get \begin{align}
    \frac{1}{\normp{I}{p}}\lVert HO\rVert_{p} &\le 
    \frac{1}{\normp{I}{p}} \left(\left(\frac{p}{1-\theta}-1\right)^{k/2} \normp{H}{(2)}\normp{I}{p/(1-\theta)}\right)\lVert O\rVert_{p}^\theta \le \left(\frac{p}{1-\theta}\right)^{k/2}\lVert H\rVert_{(2)}\left(\frac{\lVert O\rVert_{p}}{\normp{I}{p}}\right)^\theta.
\end{align}
where we explicit display the $\normp{I}{p}$ and used that $\normp{I}{p}=(\tr[I])^{1/p}$ to normalize $\normp{O}{p}$.
As in the main text, we choose
\begin{equation}
    n = \frac{1}{1-\theta} := \max\left(1,|\ln(\lVert O\rVert_{\bar p}| \right)  \le |\ln(\lVert O\rVert_{\bar p}|+1.  \label{eq:3proplast1_2}
\end{equation}
Then \begin{equation}
    \lVert HO\rVert_{\bar{p}} \le p^{k/2}\lVert H\rVert_{(2)}\lVert O\rVert_{\bar{p}} \times n^{k/2} \lVert O\rVert_{\bar{p}}^{-1/n} \le p^{k/2} \lVert H\rVert_{(2)}\lVert O \rVert_{\bar{p}} \times n^{k/2} \mathrm{e}^{|\ln \lVert O \rVert_{\bar{p}}|/n}. \label{eq:3proplast2_2}
\end{equation}
Combining (\ref{eq:3proplast1_2}) and (\ref{eq:3proplast2_2}), we obtain the advertised result.
\end{proof}

\section{Extension of submultiplicaty to tensor product ensembles}\label{app:othernorm}
Submultiplicativity also extends to special other choices of background density matrices $\rho$. We begin by recalling the following results of \cite{Beigi_2013}:
\begin{prop}\label{prop:B}
Let $\rho$ be a density matrix.  Define the norm\begin{align}
    \normp{O}{q,\rho}: =\normp{\rho^{1/2p}O\rho^{1/2p}}{p}= \left[\tr\labs{\rho^{1/2p}O\rho^{1/2p}}^p\right]^{1/p}.
\end{align}
Then the following properties hold: \begin{subequations}
    \begin{align}
        \normp{I}{q,\rho} = \tr(\rho)=1, \hspace{.5in} &\text{(proper normalization)}  \\
        \tr(\sqrt{\rho} A^\dagger \sqrt{\rho} B) \le \normp{A}{p,\rho}\normp{B}{q,\rho},\hspace{.5in} &\text{(H\"older inequality with $1/p+1/q=1$)} \\
        \lVert A\rVert_{q_\theta,\rho} \le \lVert A\rVert_{q_1,\rho}^\theta \lVert A\rVert_{q_2,\rho}^{1-\theta},\hspace{.5in} &\text{(Riesz-Thorin interpolation with $\theta/q_1+(1-\theta)/q_2=1/q_\theta$)}.
    \end{align}
\end{subequations}
\end{prop}
Using this proposition, it is straightforward to begin to extend Proposition \ref{prop:unif_subsystem} to use this generalized norm. \begin{prop}
Suppose that $[H,\rho]=0$.  Then \begin{equation}
    \normp{HO(t)}{2,\rho}^2 \le \lVert H\rVert_{2p,\rho}^2 \lVert O\rVert_{2,\rho}^{2/q} \left\lVert \rho^{-1/4}O^\dagger \sqrt{\rho} O \rho^{-1/4} \right\rVert_\infty^{1-1/q}. \label{eq:B2}
\end{equation}
\end{prop}
\begin{proof}
Using the identities in Proposition \ref{prop:B}, and temporarily just writing $O(t)=O$:
    \begin{align}
    \normp{HO}{2,\rho}^2=\tr(O^\dagger H^\dagger \sqrt{\rho} HO \sqrt{\rho}) &=\tr( \left(\rho^{-1/4}H^\dagger \sqrt{\rho} H\rho^{-1/4} \right)\sqrt{\rho}\left(\rho^{-\frac{1}{4}} O \sqrt{\rho}O^\dagger \rho^{-\frac{1}{4}} \right) \sqrt{\rho}) \notag \\
    &\le \normp{\rho^{-\frac{1}{4}}H^\dagger \sqrt{\rho} H\rho^{-\frac{1}{4}}}{p,\rho} \cdot \normp{\rho^{-\frac{1}{4}}O^\dagger \sqrt{\rho} O\rho^{-\frac{1}{4}}}{q,\rho} \notag \\
    &\le \normp{\rho^{-\frac{1}{4}}H^\dagger \sqrt{\rho} H\rho^{-\frac{1}{4}}}{p,\rho} \cdot \normp{\rho^{-\frac{1}{4}}O^\dagger \sqrt{\rho} O\rho^{-\frac{1}{4}}}{1,\rho}^{1/q} \cdot \normp{\rho^{-\frac{1}{4}}O^\dagger \sqrt{\rho} O\rho^{-\frac{1}{4}}}{\infty,\rho}^{1-1/q}\notag \\
    &\le \normp{\rho^{-\frac{1}{4}}H^\dagger \sqrt{\rho} H\rho^{-\frac{1}{4}}}{p,\rho} \cdot \tr(\sqrt{\rho} O^\dagger \sqrt{\rho} O )^{1/q} \cdot \normp{\rho^{-\frac{1}{4}}O^\dagger \sqrt{\rho} O\rho^{-\frac{1}{4}}}{\infty}^{1-1/q}.
\end{align}
In the third line we used interpolation for $\frac{1}{q}=\frac{1/q}{1}+\frac{1-1/q}{\infty}$, and in the fourth, we used that at $\infty$ it does not depends on the background anymore: $\normp{A}{\infty,\rho}=\normp{A}{\infty}$.  Lastly, when $[H,\rho]=0$, we note that for any $a,b,c$: \begin{equation}
   \rho^a O(t)^\dagger \rho^b O(t)\rho^c = \rho^a \mathrm{e}^{\mathrm{i}Ht} O^\dagger \mathrm{e}^{-\mathrm{i}Ht} \rho^b \mathrm{e}^{\mathrm{i}Ht}O\mathrm{e}^{-\mathrm{i}Ht}\rho^c = \mathrm{e}^{\mathrm{i}Ht} \rho^a O^\dagger \rho^b O\rho^c \mathrm{e}^{-\mathrm{i}Ht},
\end{equation}
and moreover we note that \begin{equation}
    \lVert A(t)\rVert_{p,\rho} = \lVert A\rVert_{p,\rho}
\end{equation}
for any $A$.  Hence we obtain (\ref{eq:B2}).
\end{proof}

We expect that this extension is most useful in the simple scenario where $\rho$ is a simple tensor product.  This naturally arises in models with a conserved charge, where one can consider the ``infinite temperature grand canonical ensemble" \begin{equation}
    \rho \propto \bigotimes_{i} \mathrm{e}^{-\mu Q_i}
\end{equation} 
where $Q_i$ represents the local charge on site $i$.  Examples correspond to spin systems with conserved total $z$-spin, or models like the Bose-Hubbard model in which total boson number is conserved.  In these models, the infinity norm $\normp{\rho^{-\frac{1}{4}}O^\dagger \sqrt{\rho} O\rho^{-\frac{1}{4}}}{\infty}$ may not be too dangerous.  Unfortunately, in a thermal ensemble with $\rho \propto \mathrm{e}^{-\beta H}$, due to the non-locality in $H$ one must be extremely careful about factors of $\mathrm{e}^{\beta H}$ when computing operator norms \cite{}.

\section{Proof of Lemma \ref{lem42}}\label{app:41}
Observe that (for simplicity we will drop the explicit time dependence in $\mathcal{L}$ and $H$) \begin{align}
    \frac{\mathrm{d}}{\mathrm{d}t} (A_0(t)|\mathcal{F}|A_0(t)) &= (A_0(t)|[\mathcal{F},\mathcal{L}]|A_0(t)) = 2\sum_{R \ge j>i\ge 0} (A_0(t)|\mathbb{Q}_j[\mathcal{F},\mathcal{L}]\mathbb{Q}_i|A_0(t)) \notag \\
    &= 2\sum_{R \ge j>i\ge 0} (j-i)(A_0(t)|\mathbb{Q}_j\mathcal{L}\mathbb{Q}_i|A_0(t)). \label{eq:trick1}
\end{align}
where the factor of $2$ comes from the fact that $j>i$ is restricted in the sum, yet the terms with $j$ and $i$ can have either number on either side of the inner product.  The antisymmetry of $\mathcal{L}$ and symmetry of $\mathcal{F}$ then allow us to say $(A|[\mathcal{F},\mathcal{L}]|A) = 2(A|\mathcal{FL}|A)$.
Our technical observation is that we can do a discrete integration by parts mimicking $\mathbb{E}[X] = \int \mathrm{d}t \; \mathbb{Q}(X>t)$ for a real-valued random variable: \begin{equation}
    \sum_{R \ge j>i\ge 0} (j-i) \mathbb{Q}_j \mathcal{L} \mathbb{Q}_i = \sum_{R \ge j>i\ge 0} \mathbb{Q}_j \mathcal{L} \sum_{k=0}^{i} \mathbb{Q}_k. \label{eq:trick2}
\end{equation}
In addition to combining these two results, we will want to split up the Liouvillian $\mathcal{L}$ using the same method of \cite{alpha_3_chenlucas}.   Let us define the following sets: \begin{equation}
    S_{q,k} := \lbrace 2^qk, 2^qk+1,\ldots, 2^{q}(k+2) - 1 \rbrace.
\end{equation}
We define \begin{equation}
    K_q = \frac{R+1}{2^q} - 1
\end{equation} to be the largest possible value of $k$ for $S_{q,k}$.
As shown in Fig.~\ref{fig:long_range_int}, at each scale $q$, these sets form ``double partitions" of the domain $[0,R]$ at every scale $0\le q \le q_*$.  
 We define \begin{equation}
    H_{q,k} := \sum_{i,j \in S_{q,k}} \indicator(i,j \text { are not contained in $S_{q',m}$ for any $q'<q$, \text{ for any } $m$}) H_{i,j},
\end{equation}
and $\mathcal{L}_{q,k}$ in the obvious way, by denoting $\mathcal{L}_{ij} = \mathrm{i}[H_{ij},\cdot]$ to be the Liouvillian associated with sites $i$ and $j$, and $\mathcal{L}_{q,k} = \mathrm{i}[H_{q,k},\cdot]$.

Now, with (\ref{eq:trick1}), \begin{align}
    \frac{\mathrm{d}}{\mathrm{d}t} (A_0(t)|\mathcal{F}|A_0(t)) &= 2 \sum_{i=0}^{R-1}\sum_{j=i+1}^R \sum_{k=0}^{i} (A_0(t)|\mathbb{Q}_j \mathcal{L} \mathbb{Q}_k |A_0(t)) =  (\sum_{q=0}^{q_*}\sum_{n=0}^{K_q}\sum_{i\in S_{q,n}})\sum_{j=i+1}^R \sum_{k=0}^{i}(A_0(t)|\mathbb{Q}_j \mathcal{L}_{q,n} \mathbb{Q}_k |A_0(t)) \notag \\
    &= \sum_{q=0}^{q_*} \sum_{n=0}^{K_q} \sum_{i\in S_{q,n}} (\sum_{j=i+1}^{2^q(n+2)})  \sum_{k=2^qn}^{i}(A_0(t)|\mathbb{Q}_j \mathcal{L}_{q,n} \mathbb{Q}_k |A_0(t)) \notag \\
    &\le \sum_{q=0}^{q_*} \sum_{n=0}^{K_q}  \sum_{i\in S_{q,n}} \left\lVert \sum_{j=i+1}^{2^q(n+2)} \mathbb{Q}_j |A_0(t)) \right\rVert_{\mathrm{F}}\cdot \left\lVert \mathcal{L}_{q,n} \sum_{k=2^qn}^{i} \mathbb{Q}_k |A_0(t)) \right\rVert_{\mathrm{F}}.
\end{align}
In the first line we absorb the factor of 2 into the double partition of $0\le i\le R-1$; in the second line, we have used the fact that $\mathcal{L}_{q,n}$ is non-vanishing only on operators with support in $S_{q,n}$, \emph{and} that if it extends a Pauli string to the right, that Pauli string has rightmost site in $S_{q,n}$. In the third line, we use the Cauchy-Schwarz inequality, keeping in mind that our inner product is already normalized to the Frobenius norm.  Now, we wish to further simplify this expression. 
To proceed, 
\begin{cor}\label{cor:42}
For any operator $O$ obeying $\lVert O\rVert_\infty \le 1$, there exists a constant $0<C<\infty$ such that \begin{equation}
    \lVert \mathcal{L}_{q,k} |O) \rVert_{\mathrm{F}} \le C \lV O\rV_{\mathrm{F}}\big(|\ln\lV O\rV_{\mathrm{F}}|+1\big) \times 2^{-q(\alpha-1)}. \label{eq:lem42}
\end{equation}
\end{cor}
\begin{proof}
Observe that for each $H_{ij}$ contained in $H_{q,k}$, we have that (if $i<j$):  $i<2^q (k+1)$ and $j \ge 2^q(k+1)$.  This means that we may apply Lemma \ref{lem:submulti}: \begin{equation}
    \lVert \mathcal{L}_{q,k} |O) \rVert_{\mathrm{F}} \le 2\times 9\times 2\mathrm{e} \sqrt{\sum_{i,j} h_{ij}^2}  \lV O\rV_{\mathrm{F}}\big(|\ln\lV O\rV_{\mathrm{F}}|+1\big)
\end{equation}
The first factor of 2 comes from considering either $HO$ or $OH$, the factor of 9 comes from triangle inequality over two sums of Paulis $\sum_{a,b=1}^3 X^aX^b$.  Now observe that \begin{equation}
    \sqrt{\sum_{i,j} h_{ij}^2} \le \sqrt{\sum_{i=-\infty}^0 \sum_{j=2^{q-1}}^\infty \frac{1}{|i-j|^\alpha}} \le \sqrt{\sum_{i=-\infty}^0 \sum_{j=2^{q-1}}^\infty \frac{1}{|i-j|^{2\alpha}}} \le const.\sqrt{ \sum_{j=2^{q-1}}^\infty \frac{1}{j^{2\alpha-1}}} \le \frac{const.}{2^{q(\alpha-1)}}.
\end{equation}
where $C_0$ and $C$ are finite constants that depend on $\alpha$. Thus we obtain (\ref{eq:lem42}).
\end{proof}

With Corollary~\ref{cor:42}, \begin{align}
    \frac{\mathrm{d}}{\mathrm{d}t} &(A_0(t)|\mathcal{F}|A_0(t)) 
    \le C \sum_{q=0}^{q_*} \sum_{n=0}^{K_q}  \sum_{i\in S_{q,n}}  \left\lVert \sum_{j=i+1}^{2^q(n+2)} \mathbb{Q}_j |A_0(t)) \right\rVert_{\mathrm{F}}\left\lVert \sum_{k=2^qn}^{i} \mathbb{Q}_k |A_0(t)) \right\rVert_{\mathrm{F}} \left(1-\ln\left\lVert \sum_{k=2^qn}^{i} \mathbb{Q}_k |A_0(t)) \right\rVert_{\mathrm{F}}\right) 2^{-q(\alpha-1)} \notag \\
     &\le C \sum_{q=0}^{q_*} \sum_{n=0}^{K_q}  \sum_{i\in S_{q,n}}  \left\lVert \sum_{j \in S_{q,n}} \mathbb{Q}_j |A_0(t)) \right\rVert_{\mathrm{F}}\left\lVert \sum_{j \in S_{q,n}} \mathbb{Q}_j |A_0(t)) \right\rVert_{\mathrm{F}} \left(1-\ln \left\lVert \sum_{j \in S_{q,n}} \mathbb{Q}_j |A_0(t)) \right\rVert_{\mathrm{F}}\right) 2^{-q(\alpha-1)} \notag \\
     &\le 2C \sum_{q=0}^{q_*} \sum_{n=0}^{K_q}  p_{q,n} \left(1-\frac{1}{2}\ln p_{q,n}\right) 2^{-q(\alpha-2)} \label{eq:mergepqn}
\end{align}
In the second line, we used that the function $x(1-\ln x)$ is increasing in the domain $x\in [0,1]$, and that the Frobenius norm of a projected operator monotonically increases if we add orthogonal projectors $\mathbb{Q}_i$s to the sum.  In the third line, we explicitly carried out the sum over $i\in S_{q,n}$, and defined \begin{equation}
    p_{q,n} := \left\lVert \sum_{j \in S_{q,n}} \mathbb{Q}_j |A_0(t)) \right\rVert_{\mathrm{F}}^2.
\end{equation}
Namely, $p_{q,n}$ is the probability that the operator has its rightmost support in $S_{q,n}$.  Note that \begin{equation}
    \sum_{n=0,2,4,\ldots} p_{q,n} \le 1, \;\;\;\; \sum_{n=1,3,5,\ldots} p_{q,n} \le 1.
\end{equation}

At this point, we first use the fact that for fixed $q$, the maximal value of the entropy like quantity (\ref{eq:mergepqn}) is attained when $p_{q,n}$ is uniform:  $p_{q,n} = 2/(K_q+1)$.  Hence, \begin{align}
    \frac{\mathrm{d}}{\mathrm{d}t} &(A_0(t)|\mathcal{F}|A_0(t)) 
    \le 2C \sum_{q=0}^{q_*} 2\left(1-\frac{1}{2}\ln \frac{2}{K_q+1}\right)2^{-q(\alpha-2)}  \le 2C\sum_{q=0}^{q_*}\left(2 + \frac{1+q_*-q}{2}\ln 2 \right)2^{-q(\alpha-2)}.
\end{align}  Evaluating this sum leads to (\ref{eq:lem42eq}).

\section{Proof of Theorem \ref{thm:pnorm_LC}}
\label{app:proof_pnorm}
As stated in the main text, we need to modify the technical proof given in \cite{alpha_3_chenlucas} in a few places to prove Theorem \ref{thm:pnorm_LC}.  Let us briefly review the proof strategy of \cite{alpha_3_chenlucas}.  As in the proof of Theorem \ref{thm:pnorm_LC}, we divide up the 1d lattice into boxes of size $2^q$ for $0\le q \lesssim \log r$.  We then observe that if (using the main text notation) $\mathbb{Q}_r|A_0(t)) \ne 0$, then, if we write out \begin{equation}
    A_0(t) = \mathrm{e}^{\mathcal{L}t}A_0 = \sum_{n=0}^\infty \frac{t^n}{n!} \mathcal{L}^n A_0
\end{equation}
as a Taylor series, we can look for the ``irreducible" $\mathcal{L}$s in any given sequence of $n$  $\mathcal{L}$s in $\mathcal{L}^n$ that ``move the operator forward":  namely move an operator from the hyperplane of $\mathbb{Q}_m$ to $\mathbb{Q}_{m+k}$ for some $k>0$.  The key idea of \cite{alpha_3_chenlucas}, using the formalism of \cite{chen2019operator}, is that (1) reducible (i.e. not irreducible) $\mathcal{L}$s can be re-exponentiated to form a new unitary transformation, and (2) we only need to count ``irreducible" $\mathcal{L}$s in sequences that have sufficiently many steps forward at the corresponding scale $q$ (as if we pull down a small number of irreducible steps at scale $q^\prime$, we can just re-exponentiate $\mathcal{L}$s at this scale).

We will want to perform each of these 2 steps a little differently.  However, to motivate why, it will help to invert the logic and first think about what we wish to accomplish during the second step.  At the longest scale $q_*$, we only want to count \emph{one} irreducible path.  Let $\mathcal{L}_{q_*}$ contains all terms in $\mathcal{L}$ corresponding to a commutator $[H_{mn},\cdots]$ at the longest scale, $\mathcal{L}_{-q_*} = \mathcal{L}-\mathcal{L}_{q_*}$, and $\mathcal{L}_{q_*,j}$ contains all couplings in $\mathcal{L}_{q_*}$ that end on site $j$.   Using the Duhamel identity, we write \begin{equation}
    \mathrm{e}^{\mathcal{L}t} = \mathrm{e}^{\mathcal{L}_{-q_*}t} + \sum_{j>r/2} \int\limits_0^t \mathrm{d}s \mathrm{e}^{\mathcal{L}(t-s)} \mathcal{L}_{q_*,j} \mathrm{e}^{\mathcal{L}_{-q_*}s}.
\end{equation}
Now consider using the triangle inequality, which holds for all $p$-norms, to obtain \begin{equation}
    \lVert \mathbb{Q}_r \mathrm{e}^{\mathcal{L}t} |\mathcal{O}_0) \rVert_{\bar p} \le \lVert \mathbb{Q}_r \mathrm{e}^{\mathcal{L}_{-q_*}t} |\mathcal{O}_0) \rVert_{\bar p} + \sum_{j>r/2}\int\limits_0^t \mathrm{d}s \lVert \mathbb{Q}_r \mathrm{e}^{\mathcal{L}(t-s)} \mathcal{L}_{q_*,j} \mathrm{e}^{\mathcal{L}_{-q_*}s} |\mathcal{O}_0) \rVert_{\bar p}.
\end{equation}
The first term in this inequality would correspond to all paths that make it from 0 to $r$ without ultra-long couplings, and can be bounded separately (as we will describe below).  The second term corresponds to all of the paths that invoked a long coupling to get to a site $j$ more than halfway to $r$.  We would be tempted to perform the following simplifications: \begin{equation}
    \lVert \mathbb{Q}_r \mathrm{e}^{\mathcal{L}(t-s)} \mathcal{L}_{q_*,j} \mathrm{e}^{\mathcal{L}_{-q_*}s} |\mathcal{O}_0) \rVert_{\bar p} \lesssim \lVert \mathcal{L}_{q_*,j} \mathrm{e}^{\mathcal{L}_{-q_*}s} |\mathcal{O}_0) \rVert_{\bar p} \lesssim \lVert H_{q_*,j}\rVert_{\bar p} \lVert\mathrm{e}^{\mathcal{L}_{-q_*}s} |\mathcal{O}_0) \rVert_{\bar p} . \label{eq:Qrpnormtry}
\end{equation}
The first step is very mild, contributing at worst an O(1) factor.  The second step, however, can only be used \textit{once} as we have emphasized for the Frobenius light cone~\eqref{eq:only_once}: 
To peel off $\lVert HO\rVert_{\bar p} \sim \lVert H\rVert_{(2)}\lVert O\rVert_{\bar p} $, submultiplicativity (Proposition \ref{prop:submul_pnorm}) requires a constraint $\lVert O\rVert_\infty \le 1$.  For (\ref{eq:Qrpnormtry}) this criterion is satisfied, but it will not be whenever there are 2 or more irreducible couplings in a sequence. 
Our remedy is the following. For illustration, consider growing an operator $\CO$ by $\CL $, and for simplicity assume the operator is supported entirely $\le 0$ and $\CL$ is bipartite $j>0>i$. Then, uniform smoothness implies
\begin{align}
\normp{\CL      |\mathcal{O})}{\bar{p}}^2 &= \lnormp{\sum_{j>0>i} \CL_{ij}      |\mathcal{O}_{\le 0})}{\bar{p}}^2 \\
&\le (p-1) \sum_{j>0} \lnormp{ \sum_{0>i} \CL_{ij}      |\mathcal{O}_{\le 0})}{\bar{p}}^2\\
&\le (p-1) \sum_{j>0}  \lnorm{2\sum_{0>i} H_{ij}}^2    \cdot  \lnormp{|\mathcal{O}_{\le 0})}{\bar{p}}^2.
\end{align}

In words, the sum over the right site $j$ leads to a sum of squares in operator weight (as in Frobenius light cones) while the sum over left site $i$ is bounded additively as in the Lieb-Robinson light cone. What we gain from this sacrifice is that this triggers a recursion for the (normalized) p-norm. We will see that this leads to a light cone of $t\sim r^{\min(1,\alpha-3/2)}$, in between the Frobenius light cone of $t\sim r^{\min(1,\alpha-1)}$ and Lieb-Robinson light cone $t\sim r^{\min(1,\alpha-2)}$. 

Having made this observation, we can now essentially reproduce the proof used in \cite{alpha_3_chenlucas}, except that we need to modify a few of the constants as well as the equivalence classes used in that proof. Disect the interactions of scale q and regroup the terms sharing the same site $j$ for $i<j$ (i.e. $j$ is farther from the origin). Labeled by sites $r\ge j \ge 0$, these are the building blocks we will use.
\begin{align}
    \sum_k H_{q,k} = \sum_k \sum_{(i,j)\in H_{q,k}} H_{ij} = \sum_j \L[ \sum_k \sum_{i<j, (i,j)\in H_{q,k}} H_{ij}\R] =: \sum_j H_{q,j}.
\end{align}

For any sequence, we can define and read out its \textit{forward sequence} by the following recursive algorithm.
Suppose $\CL_{\beta_n}\cdots \CL_{\beta_1}$ has forward sequence $f_n= (\CL_{\beta_m}, \cdots, \CL_{\beta_1})$. 
Then for $\CL_{\beta_{n+1}}\cdots \CL_{\beta_1}$, it has forward sequence
\begin{align}
    f_{n+1}=  \begin{cases}
    (\beta_{n+1}, f_n ),\ &\textrm{if}\ j_{n+1}\ge j_{m}+1 \\
    f_n \ &\textrm{else}.
    \end{cases}
\end{align}
where $\beta_n\equiv (q_n,k_n,j_n)$. 
In words, the $j_{n+1}\ge j_{m}+1$ condition says a hop is forward if exceeds the farthest so far. 

 For each forward sequence, we can isolate terms of some scale $q$ to form a \textit{q-forward-subsequence}. Such a sequence is \textit{long} if has at least
\begin{align}
    N_q:= \L\lceil \dfrac{1}{2}\frac{2^{-q(\alpha-5/2)\frac{2}{3}}}{ \displaystyle \sum_{q'=0}^{q_*} 2^{-q'(\alpha-5/2)\frac{2}{3}} }\frac{r}{2^{q+1}}\R\rceil\label{eq:Nq}
\end{align}
terms. Note the size of a scale-$q$ block is $2^{q+1}$.
We then define characteristic function $\chi_q$ for each scale
\begin{align}
    \chi_q \CL_{\beta_n}\cdots \CL_{\beta_1}:=
    \begin{cases}
     \CL_{\beta_n}\cdots \CL_{\beta_1} &\textrm{if it contains a long q-forward subsequence }\\
     0 &\textrm{else},
    \end{cases}
\end{align}
and for each forward subsequence $\beta$
\begin{align}
    \chi_{\beta} \CL_{\beta_n}\cdots \CL_{\beta_1}:=
    \begin{cases}
     \CL_{\beta_n}\cdots \CL_{\beta_1} &\textrm{if it contains the forward subsequence $\beta$ }\\
     0 &\textrm{else}.
    \end{cases}
\end{align}
We denote with $\mathcal{F}_q$ the set of long $q$-forward (sub)sequences.  The following combinatorial proposition will now prove useful (See, e.g.~\cite{aigner_2010}, for its context in combinatorics):
\begin{prop}\label{fact:gen_incl_excl}
Let $\mathcal{F}_q$ denote the set of all long forward $q$-sequences of length $\ge N_q$.  Then \begin{align}\label{eq:propD1}
    \chi_q = \sum_{k=N_q}^{r} \L(\sum^{k}_{p=N_q} (-1)^{k-p} \binom{k}{p}\R) \sum_{\beta \in \CF_q:\labs{\beta}=k} \chi_{\beta}.
\end{align}
\end{prop}
\begin{proof}
Suppose that we have a  sequence $\mathcal{S} := \mathcal{L}_{\beta_n}\cdots \mathcal{L}_{\beta_1}$ First note that $\chi_q$ clearly annihilates anything which does not have a long $q$-forward subsequence, since $\chi_\beta \mathcal{S}=0$ for any $\beta\in\mathcal{F}_q$.  

Hence we can assume that $\mathcal{S}$ does have a long $q$-forward subsequence $\beta$.  Consider what happens at length $N_q+k$.  For simplicity, let $\tilde\chi_q$ denote the right hand side of (\ref{eq:propD1}).  In general, we will have \begin{equation}
    \tilde\chi_q \mathcal{S} = \mathcal{N} \cdot \mathcal{S},
\end{equation}
where the integer $\mathcal{N}$ counts the number of $\chi_\beta$s, weighted by their appropriate prefactor in (\ref{eq:propD1}), which do not annihilate $\mathcal{S}$.  We can straightforwardly calculate \begin{align}
    \mathcal{N} &= \sum_{r=0}^k \left(\begin{array}{c} N_q+k \\ N_q+r\end{array}\right) \times \L(\sum^{r}_{p=0} (-1)^{r-p} \binom{N_q+r}{N_q+p}\R) = \sum_{r=0}^k\sum^{r}_{p=0}(-1)^{r-p} \frac{(N_q+k)!}{(k-r)!(r-p)!(N_q+p)!} \notag \\
    &= \sum_{p=0}^k\sum^{k-p}_{q=0}(-1)^{q} \frac{(k-p)!}{(k-p-q)!q!} \frac{(N_q+k)!}{(k-p)!(N_q+p)!}
\end{align}
where in the second line, we switched variables to $r=p+q$ in the sum.  Clearly, the sum over $q$ vanishes unless it only runs over $q=0$; namely, $p=k$.  We conclude that only the $k=p$ term above is non-zero, which immediately leads to $\mathcal{N}=1$.  Hence $\tilde\chi_q\mathcal{S}=\mathcal{S}$, which implies $\tilde\chi_q=\chi_q$ and thus (\ref{eq:propD1}).
\end{proof}

We will use the following simple bound in what follows: \begin{equation}
    \sum^{k}_{p=q} (-1)^{k-p} \binom{k}{p}\le 2^k. \label{eq:binomsimple}
\end{equation}
Next, we let $\mathrm{\Delta}^\ell(t)$ denotes the $\ell$-simplex: \begin{equation}
\mathrm{\Delta}^\ell(t) := \lbrace (t_1,\ldots, t_\ell) \in [0,t]^\ell : t_1 \le t_2 \le \cdots \le t_\ell \rbrace.
\end{equation}
The following lemma helps us to use the indicator functions above to group all sequences in $\mathrm{e}^{\mathcal{L}t}$ into ``irreducible" sequences: 
\begin{lem}\label{lem:resum}
\begin{align}
    \chi_\beta \e^{\CL t}|A_0) = \int\limits_{\mathrm{\Delta}^\ell(t)} \mathrm{d}t_\ell\cdots \mathrm{d}t_1
    \e^{\CL(t-t_\ell)} \CL_{\beta_\ell} \e^{\CL^\beta_\ell(t_\ell - t_{\ell-1})}\CL_{\beta_{\ell-1}} \e^{\CL^\beta_{\ell-1}(t_{\ell-1} - t_{\ell-2})} \cdots \CL_{\beta_1} \e^{\CL^\beta_1t_1} |A_0)
\end{align}
where $\ell = \ell(\beta)$, \begin{equation}
\mathcal{L}^\beta_p := \mathcal{L} - \sum_{\lambda \in Y_p(\beta)} \mathcal{L}_\lambda \label{eq:qLbetap}
\end{equation}
with the sets $Y_p^\beta$ corresponding to the forbidden terms\begin{align}
 Y^q_p(\beta) &:=  \lbrace (q^\prime, k, j): j\ge j(\beta_{p-1}) \rbrace, 
\end{align}
which would change the forwardness of the sequence, at any scale.  $j(\beta_p)$ denotes the right site of coupling $\beta_p$.
\end{lem}
The proof of this lemma is found in \cite{alpha_3_chenlucas}. In words, the intermediate unitaries $\mathcal{L}_p^\beta$ contain the terms (at any scale) which cannot possibly render the coupling $\mathcal{L}_{\beta_p}$ a non-forward coupling.  The following lemma allows us to then bound the $p$-norm of the resulting sequences:
\begin{lem}
\begin{align}
     \lnormp{\sum_{\beta\in \CF_q, \labs{\beta}=\ell} \chi_\beta \e^{\CL t}|A_0)}{p} \le \frac{\L(2\sqrt{C_pr} \norm{H_{q,i}}t\R)^{\ell}}{\ell!\sqrt{\ell!}}  \normp{A_0}{p}. 
\end{align}
\end{lem}
\begin{proof}
The $q$-forward sequence can be enumerated by the right-most sites of the $q$-scale couplings: we denote these as $r\ge j_\ell>\cdots>j_1\ge 1$, using the previous lemma:
\begin{align}
      \sum_{\beta\in \CF_q, \labs{\beta}=\ell} &\chi_\beta \e^{\CL t}|A_0) 
     =  \sum_{\beta\in \CF_q, \labs{\beta}=\ell}\int\limits_{\mathrm{\Delta}^\ell(t)} \mathrm{d}t_\ell\cdots \mathrm{d}t_1
    \e^{\CL(t-t_\ell)} \CL_{\beta_\ell} \e^{\CL^\beta_\ell(t_\ell - t_{\ell-1})}\CL_{\beta_{\ell-1}} \e^{\CL^\beta_{\ell-1}(t_{\ell-1} - t_{\ell-2})} \cdots \CL_{\beta_1} \e^{\CL^\beta_1t_1} |A_0) \notag \\
    &=\sum_{j_\ell>\cdots> j_1} \int\limits_{\mathrm{\Delta}^\ell(t)} \mathrm{d}t_\ell\cdots \mathrm{d}t_1
    \e^{\CL(t-t_\ell)} \CL_{j_\ell} \e^{\CL_{(j_\ell)}(t_\ell - t_{\ell-1})}\CL_{j_{\ell-1}} \e^{\CL_{(j_{\ell-1})}(t_{\ell-1} - t_{\ell-2})} \cdots \CL_{j_1} \e^{\CL_{(j_1)}t_1} |A_0) \notag \\
    &= \int\limits_{\mathrm{\Delta}^\ell(t)} \mathrm{d}t_\ell\cdots \mathrm{d}t_1
    \e^{\CL(t-t_\ell)} \sum_{j_\ell} \CL_{j_\ell} \e^{\CL_{(j_\ell)}(t_\ell - t_{\ell-1})}\sum_{j_{\ell-1}<j_{\ell} } \CL_{j_{\ell-1}} \e^{\CL_{(j_{\ell-1})}(t_{\ell-1} - t_{\ell-2})} \cdots \sum_{j_1<j_2}\CL_{j_1} \e^{\CL_{(j_1)}t_1} |A_0).
\end{align}
We use notation  $\CL_{(i_\ell)}$ to emphasize it is independent of the previous $i_{\ell-1},\cdots ,i_{1}$(Lemma~\ref{lem:resum}), which then allows us moves the sums inside of the integrals in the third equality. This nested sum then conveniently allows us to use the invariance of $p$-norms under unitary evolution to remove all but the ``irreducible" steps $\mathcal{L}_{\beta_p}$ from our bound:
\begin{align}
 &\lnormp{ \e^{\CL(t-t_\ell)} \sum_{j_\ell} \CL_{j_\ell} \e^{\CL_{(j_\ell)}(t_\ell - t_{\ell-1})}\sum_{j_{\ell-1}<j_{\ell} } \CL_{j_{\ell-1}} \e^{\CL_{(j_{\ell-1})}(t_{\ell-1} - t_{\ell-2})} \cdots \sum_{j_1<j_2}\CL_{j_1} \e^{\CL_{(j_1)}t_1} |A_0)}{p}^2 \notag \\
    &\le  C_p \sum_{j_\ell} \norm{2H_{j_\ell}}^2_\infty \lnormp{\e^{\CL_{(j_\ell)}(t_\ell - t_{\ell-1})}\sum_{j_{\ell-1}<i_{\ell} } \CL_{j_{\ell-1}} \e^{\CL_{(j_{\ell-1})}(t_{\ell-1} - t_{\ell-2})} \cdots \sum_{j_1<j_2}\CL_{j_1} \e^{\CL_{(j_1)}t_1} |A_0)}{p}^2 \notag \\
    & \le (4C_p)^\ell \sum_{j_\ell>\cdots> j_1} \norm{ H_{j_{\ell}} }^2\cdots \norm{ H_{j_{1}} }^2 \normp{A_0}{p}^2
    \le  \frac{(4C_pr \norm{H_{q,j}}^2)^{\ell}}{\ell!} \normp{A_0}{p}^2. 
\end{align}
The first inequality uses uniform smoothness for each term $\CL_{j_\ell}\cdots$, which is non-trivial on site $j_\ell$ and trivial beyond.  Indeed, if we start with the right-most possible $j_\ell$, we know that all other terms in the sum are trivial on site $j_\ell$, so uniform smoothness applies; we simply repeat this argument until we have summed over all $j_\ell$. The second line repeats this for $i_{\ell-1}\cdots i_1$. In the last line, we used the combinatorial bound \begin{equation}
    \sum_{r\ge i_\ell>\cdots> i_1 \ge 1} \le \frac{r^{\ell}}{\ell!}.
\end{equation}
Finally, we use one final triangle inequality to state that \begin{equation}
    \norm{\sum_{\beta\in \CF_q, \labs{\beta}=\ell} \chi_\beta \e^{\CL t}|A_0)}_p \le \sqrt{\frac{(4C_pr \norm{H_{q,j}}^2)^{\ell}}{\ell!} \normp{A_0}{p}^2} \int\limits_{\Delta^\ell(t)}\mathrm{d}t_\ell \cdots \mathrm{d}t_1 \; 1 = \frac{(2\sqrt{C_pr}t\lVert H_{q,j}\rVert)^\ell}{\ell!^{3/2}}
\end{equation}
which completes the proof.
\end{proof}
We now generalize this result to sequences that have $q$-forward subsequences at multiple scales $q$, which we index by set $Z \subset (0, \cdots , q^*)$.  As before, we label with $\mathcal{F}_Z$ the set of all such (sub)sequences (keep in mind that the ordering of terms at different scales is important to keep track of), and we define \begin{equation}
    \chi_Z := \prod_{q\in Z}\chi_q.
\end{equation}
\begin{lem}\label{lem:D4}
\begin{align}
     \lnormp{\sum_{\beta\in \CF_Z, \labs{\beta_q}=\ell_q} \chi_\beta \e^{\CL t}|A_0)}{p} 
     &\le \normp{A_0}{p} \prod_{q\in Z} \frac{\L(2\sqrt{C_pr} \norm{H_{q,j}}t\R)^{\ell_q}}{\ell_q!\sqrt{\ell_q!}} \label{eq:D4}
\end{align}
\end{lem}
\begin{proof}
The proof is nearly identical, but is slightly more tedious due to the multiple scales.  Let us first bound the number of ways that different scales can weave through each other, which is bounded by the multinomial coefficient \begin{equation}
   \sum_{\beta\in \CF_Z, \labs{\beta_q}=\ell_q}1 \le  \ell! \prod_{q\in Z}\frac{1}{\ell_q!}
\end{equation}where \begin{equation}
    \ell := \sum_{q\in Z}\ell_q.
\end{equation}For each sequence in $Z$, we obtain a factor of 
\begin{align}
    \lnormp{ \chi_\beta \e^{\CL t}|A_0)}{p}^2 \le  (4C_p)^\ell \sum_{j_\ell>\cdots> j_1} \norm{ H_{j_{\ell}} }^2\cdots \norm{ H_{j_{1}} }^2 \normp{A_0}{p}^2
    \le \frac{(4C_pr)^{\ell}}{\ell!} \normp{A_0}{p}^2 \prod_{j=1}^\ell \norm{H_{q,j}}^2.
\end{align}
Following the last steps of the prior lemma, we find that  \begin{align}
    \lnormp{\sum_{\beta\in \CF_Z, \labs{\beta_q}=\ell_q} \chi_\beta \e^{\CL t}|A_0)}{p} &\le \lVert A_0\rVert_p
    \sum_{\beta \in F_Z, |\beta_q|=\ell_q} \frac{t^\ell}{\ell!} \le \frac{1}{\sqrt{\ell!}}\lVert A_0\rVert_p \prod_{q\in Z}\frac{(2\sqrt{C_pr}\lVert H_{q,j}\rVert t)^{\ell_q}}{\ell_q!} \sqrt{\frac{(4C_pr)^\ell}{\ell!}} \prod_{j=1}^\ell \lVert H_{q,j}\rVert.
\end{align} To finish the proof, we simply use the loose bound \begin{equation}
    \ell! \ge \prod_{q\in Z} \ell_q!,
\end{equation}
which leads to (\ref{eq:D4}).
\end{proof}
At this point, we invoke the inclusion-exclusion of different scales found in \cite{alpha_3_chenlucas}: \begin{equation}
     \lnormp{\mathbb{P}_r \e^{\CL t}|A_0)}{p} = \lnormp{\mathbb{P}_r \sum_{Z\ne \emptyset} (-1)^{|Z|-1}\chi_Z \e^{\CL t}|A_0)}{p} \le \sum_{Z\ne \emptyset} \lnormp{\mathbb{P}_r \chi_Z \e^{\CL t}|A_0)}{p}.
\end{equation}
Together with Proposition \ref{fact:gen_incl_excl} (which immediately generalizes to multi-scale $Z$), with \begin{equation}
    N_Z := \sum_{q\in Z}N_q,
\end{equation}
we find that 
\begin{align}
    \lnormp{\mathbb{P}_r \chi_Z \e^{\CL t}|A_0)}{p} &= \lnormp{\sum_{\ell=N_Z}^r \left(\sum_{p=N_Z}^\ell (-1)^{\ell-p}\left(\begin{array}{c} \ell \\ p \end{array}\right)\right) \sum_{\beta \in \mathcal{F}_Z: |\beta|=\ell}\chi_\beta\e^{\CL t}|A_0)}{p} \notag \\
    &\le \sum_{\ell=N_Z}^r 2^\ell \sum_{\beta \in \mathcal{F}_Z: |\beta|=\ell}\lnormp{\chi_\beta\e^{\CL t}|A_0)}{p} \le \sum_Z \prod_{q\in Z}\sum_{\ell_q=N_q}^\infty \frac{(4\sqrt{C_p r}\lVert H_{q,j}\rVert t)^{\ell_q}}{\ell_q!^{3/2}}. \label{eq:D38}
\end{align}
In the second line we used (\ref{eq:binomsimple}), followed by Lemma \ref{lem:D4}.  

Now, let us suppose that \begin{equation}
    1\ge \frac{4\sqrt{C_pr}\lVert H_{q,i}\rVert t}{N_q^{3/2}} \label{eq:D39}
\end{equation}
for every scale $q$.  Assuming this inequality (which will fix the values of $t$ for which our bound is valid), then (\ref{eq:D38}) becomes \begin{equation}
    \lnormp{\mathbb{P}_r \e^{\CL t}|A_0)}{p} \le 2 \sum_Z \prod_{q\in Z}\frac{(4\sqrt{C_p r}\lVert H_{q,j}\rVert t)^{N_q}}{N_q!^{3/2}} \le -1 + \exp\left(2\sum_q \frac{(4\sqrt{C_p r}\lVert H_{q,j}\rVert t)^{\ell_q}}{\ell_q!^{3/2}}\right).
\end{equation}

It is useful to determine the first value $q_1$ at which a long $q$-forward path has a single coupling: $N_q=1$ for $q\ge q_1$.   (Recall $N_q$'s definition in (\ref{eq:Nq}).)   This occurs when \begin{equation}
\frac{M}{r}\ge \frac{1}{4}\frac{1}{(2^{q_1})^{(\alpha-1)\frac{2}{3}}}, \label{eq:DM}
\end{equation}
where we defined \begin{equation}
M = \sum_{q=0}^{q_*} 2^{-q(\alpha-5/2)\frac{2}{3} }.
\end{equation}
Then, noting that for any $j$, there exists a constant $0<c_1^\prime<\infty$ such that \begin{equation}
    \lVert H_{q,j}\rVert \le \sum_{i=0}^{j-1} \mathbb{I}(H_{ij} \in H_{q,k})\lVert H_{ij}\rVert  \le \frac{c_1^\prime}{4 (2^q)^{\alpha-1}},
\end{equation}
we find that there exist constants $0<c_2^\prime,c_2<\infty$ such that
\begin{align}
    \sum_{q} 2\frac{\L(4\sqrt{C_pr} \norm{H_{q,i}}t\R)^{N_q}}{N_q!^{3/2}} &\le \sum_{q} \L( \frac{c'_2\sqrt{pr}t}{2^{q(\alpha-1)} N_q^{3/2}} \R)^{N_q}\notag \\
    &\le \sum_{q=0}^{q_1-1} \L( \frac{c_2\sqrt{p}tM^{3/2}}{r} \R)^{N_q} + c_3\sqrt{pr}t \sum_{q=q_1}^{q_*} \frac{1}{2^{q(\alpha-1)}}.
\end{align}

We now analyze this sum for different ranges of $\alpha$.

\textbf{Case: $\alpha>5/2$.}
First, we note that $M$ does not depend on $r$, since even as $r\rightarrow \infty$ the sum in (\ref{eq:DM}) is convergent.   Secondly, we observe that $N_1>N_2>\cdots > N_{q_*}$, which follows from (\ref{eq:Nq}) and the fact that $N_q$ (before the floor function) changes by a factor of at least $2^{\frac{2}{3}(\alpha-1)}>2$ each time.  These two inequalities imply that there exist constants $0<c_{4,5,6}<\infty$ which do not depend on $r$ such that 
\begin{align}
    \sum_{q=0}^{q_1-1} \L( \frac{c_2\sqrt{p}tM^{3/2}}{r} \R)^{N_q} + c_3\sqrt{pr}t \sum_{q=q_1}^{q_*} \frac{1}{2^{q(\alpha-1)}} &\le \sum_{n=2}^{\infty} \L( \frac{c_2\sqrt{p}tM^{3/2}}{r} \R)^{n} + \frac{c_3\sqrt{pr}t}{2^{q_*(\alpha-1)}}\notag \\ 
    &\le c_4 \frac{\sqrt{p}t}{r-c_5\sqrt{p}t} + \frac{c_6\sqrt{p}t}{r}.
\end{align}
Observe that (\ref{eq:D39}) holds so long as \begin{equation}
    t \le \frac{N_q^{3/2}}{4\sqrt{pr}}2^{q(\alpha-1)} \le k_1^\prime \frac{r^{3/2}}{2^{q(\alpha-1)}\sqrt{pr}(4M)^{3/2}}2^{q(\alpha-1)} \le \frac{k_1r}{\sqrt{p}},
\end{equation}
for constants $0<k_1^\prime,k_1<\infty$ that do not depend on $r$.

\textbf{Case: $3/2 \le \alpha < 5/2$}.  Now we find that (for constants $0<c_{7}<\infty$ and an integer $m$ independent of $r$): \begin{equation}
M  = \sum_{q=0}^{q_*} 2^{-q(\alpha-5/2)\frac{2}{3} } = c_7 r^{(5/2-\alpha)\frac{2}{3}}, 
\end{equation}
and \begin{equation}
    q_1 = q_* - m.
\end{equation}
Then note that $
N_1 > N_4 > N_7\cdots ,$
because the argument of (\ref{eq:Nq}) now only varies by $2^{(\alpha-1 )2/3} \ge 2^{1/3}$ each time $q$ varies by 1.
Hence, we obtain (for constants $0<c_{8,9,10}<\infty$ independent of $r$):
\begin{align}
    \sum_{q=0}^{q_1-1} \L( \frac{c_2\sqrt{p}tM^{3/2}}{r} \R)^{N_q} + c_3\sqrt{pr}t \sum_{q=q_1}^{q_*} \frac{1}{2^{q(\alpha-1)}} &\le 3 \sum_{n=2}^{\infty} \L( \frac{c_2\sqrt{p}tM^{3/2}}{r} \R)^{n} + c_3\sqrt{pr}t \sum_{q=q_1}^{q_*} \frac{1}{2^{q(\alpha-1)}} \notag \\ 
    &\le \frac{c_8 \sqrt{p}t}{r^{\alpha-3/2} - c_9 \sqrt{p}t} + \frac{c_{10}\sqrt{p}t}{r^{\alpha-3/2}}.
\end{align}
Observe that (\ref{eq:D39}) holds so long as \begin{equation}
    t \le \frac{N_{q}^{3/2} }{4\sqrt{pr}}2^{q(\alpha-3/2)} \le \frac{r^{3/2}2^{q(5/2-\alpha)}2^{-3q/2}}{\sqrt{pr}M^{3/2}}2^{q(\alpha-1)} =k_2 \frac{r}{\sqrt{p}}{r^{5/2-\alpha}}= k_2 \frac{r^{\alpha-3/2}}{\sqrt{p}}.
\end{equation}
for $0<k_2<\infty$.

\textbf{Case : $\alpha=5/2$.}  We obtain \begin{equation}
    M=q_*+1.
\end{equation}
For $r$-independent constant $0<c_{11}<\infty$, \begin{equation}
    2^{q_1} = c_{11} \frac{r}{\log_2 r}.
\end{equation}
Again we have $N_1>N_2>\cdots > N_{q_1-1}>1$, for the same reason as when $\alpha>5/2$.  Hence for $0<c_{12,13,14}<\infty$,
\begin{align}
    \sum_{q=0}^{q_1-1} \L( \frac{c_2\sqrt{p}tM^{3/2}}{r} \R)^{N_q} + c_3\sqrt{pr}t \sum_{q=q_1}^{q_*} \frac{1}{2^{q(\alpha-1)}}. &\le  \sum_{n=2}^{\infty} \L( \frac{c_2\sqrt{p}t\ln(r)^{3/2}}{r} \R)^{n} + c_3\sqrt{pr}t \sum_{q=q_1}^{q_*} \frac{1}{2^{q(\alpha-1)}}\\ 
    &\le \frac{c_{12}\sqrt{p}t}{r\ln^{-3/2} r - c_{13}\sqrt{p}t} + \frac{c_{14} \sqrt{p}t \ln r}{r}.
\end{align}
Lastly, observe that (\ref{eq:D39}) is satisfied so long as \begin{equation}
    t \le \frac{N_q^{3/2}}{4\sqrt{pr}}2^{q} \le k_3 \frac{r^{3/2}}{\ln^{3/2}r} \frac{2^{q -q }}{\sqrt{pr}} = \frac{k_3}{\sqrt{p}} \frac{r}{\ln^{3/2}r}
\end{equation}for some constant $0<k_3<\infty$.

Now let us summarize our results. Recall the definition of $\mathcal{R}(r)$ in (\ref{eq:mathcalR}).  We have shown that for constants $0<k_{4,5}<\infty$, for times $t<k_4 \mathcal{R}(r)$,
\begin{align}
    \frac{ \normp{\BP_r \e^{\CL t}|A_0)}{p}}{\normp{A_0}{p}} \le \frac{k_5t}{\mathcal{R}(r)}.
\end{align}
This completes the proof.
\section{Proof of Proposition \ref{thm5}}\label{app:5}

\emph{Step 1:} We begin by developing the ``super-density matrix" picture of operator growth, following \cite{Lucas:2019aif}.  If $\mathcal{H}$ is the quantum mechanical Hilbert space, and (with a slight abuse of notation) $\mathcal{H}\otimes \mathcal{H}$ is a Hilbert space of all normalizable operators on $\mathcal{H}$ (in what follows, we will restrict to Hermitian operators), the super-density matrix is a normalizable element of $(\mathcal{H}\otimes \mathcal{H})\otimes(\mathcal{H}\otimes \mathcal{H})$.   It is easier to visualize with bra-ket notation:  using the operator ket $|A)$ introduced above, the pure super-density matrix corresponding to $A$ becomes $|A)(A|$.   It is straightforward to build a good basis for super-operators. Using the fact that the space of Hermitian operators acting on a single qubit is spanned by the orthonormal basis
\begin{align}
  |X^a) \;\;\; (a=0,1,2,3) :=  |I),|X),|Y),|Z) \in \BR^{1+3} \label{eq:BR13}
\end{align}
endowed with the canonical inner product in $\BR^{1+3}$, a basis for super-operators on a single qubit is evidently $|X^a)(X^b|$ for $a,b=0,1,2,3$.   The standard tensor product between lattice sites then allows us to build up a good basis for our super-operator space.

Observe that the Frobenius norm is now simply (up to square) the super-operator trace \begin{equation}
    \mathrm{Tr} [|A)(A| ] := (A|A) = \lVert A \rVert_{\mathrm{F}}^2.
\end{equation}
The latter equality follows from (\ref{eq:innerproduct}).  As a consequence, evaluating the Frobenius norms of projected operators in this super-operator language will be particularly simple -- we simply pick out the basis operators which we wish to keep, and sum the coefficients of the diagonal elements of the pure super-density matrix.  In particular, given the growing operator $X_0(t)$, we can define the probabilities $p_S(t)$ (for subsets $S\subset\Lambda$) as (recall (\ref{eq:pS})) \begin{equation}
    |X_0(t))(X_0(t)| := \sum_{S\subset \Lambda} p_S(t) |X_0(t)_S) (X_0(t)_S| + \text{off-diagonal terms}.
\end{equation}
Much of the proof that follows will amount to bounding (sums of) $p_S(t)$ throughout the protocol. 
To avoid confusing the super-operator spaces from operators, we here define the adjoint operation acting on \textit{operators} $\mathcal{H}\otimes \mathcal{H} \rightarrow \mathcal{H}\otimes \mathcal{H}$  \begin{equation}
    \mathrm{Adj}_{\mathrm{i}A} (X):= |\mathrm{i}[A,X]).
\end{equation}
and the conjugation operation \begin{equation}
    \mathrm{Conj}_B | X) := | B^\dagger XB ).
\end{equation}
We define the super-depolarizing channel, acting on \textit{super-operators} $(\mathcal{H}\otimes \mathcal{H})\otimes(\mathcal{H}\otimes \mathcal{H}) \rightarrow (\mathcal{H}\otimes \mathcal{H})\otimes(\mathcal{H}\otimes \mathcal{H})$ as \begin{align}
    \mathcal{D} := \mathbb{E}_D \left[ \left(\mathrm{Conj}_D\right)(\cdot)\left(\mathrm{Conj}_D\right)^\dagger \right].
\end{align}
Recall the depolarizer $D$ is a random tensor product unitary defined in the main text.  To give a concrete example:
\begin{align}
    \mathcal{D}[|X)(Y|] = \BE_D |D^\dagger X D)(D^\dagger Y D|.
\end{align}
We define the $q$-scale growth channel as \begin{equation}
    \mathcal{V}_q := \mathbb{E}_J \left[ \left(\mathrm{Conj}_{V_q}\right)(\cdot)\left(\mathrm{Conj}_{V_q}\right)^\dagger \right].
\end{equation}

Expectation values are taken over the random $D_x$ in (\ref{eq:Dx}), and random $J_{xy}$ in (\ref{eq:HZZq}), respectively.  We thus see that our (averaged) protocol as a ``quantum super-channel", defined via the analogue of (\ref{eq:induction}): \begin{equation}
    \mathcal{M}_q := \mathcal{M}_{q-1} \mathcal{D} \mathcal{V}_q \mathcal{M}_{q-1}.  \label{eq:induction2}
    \end{equation}

\emph{Step 2:} We now turn to the analysis of the super-depolarizing channel $\mathcal{D}$.   Let us define the ``maximally mixed non-trivial operator" (on a single site) \begin{equation}
    \mu := \frac{|X)(X| + |Y)(Y|+ |Z)(Z|}{3}.
\end{equation}
In a nutshell, the essence of this proof is that all we need to keep track of is whether or not each site has a $\mu$ or the trivial operator\begin{equation}
    \mathcal{I} = |I)(I| 
\end{equation}
on it: namely, we will show that after each step of the random protocol, \begin{equation}
    \CD\mathcal{V}_q |X_0)(X_0| = \sum_{S\subset\Lambda} \rho_S \bigotimes_{i\in S}\mu_i \bigotimes_{j\in S^c}I_j. \label{eq:Vqmu}
\end{equation}  

The depolarizing channel is what makes this simplification possible, as formalized via the following:

\begin{prop}\label{prop:depolarize}
For any $a,b=1,2,3$ (and on every site independently),
\begin{subequations}\label{eq:depolarize}\begin{align}
    \mathcal{D} |I)(I| &= |I)(I|,\\
    \mathcal{D} |I)(X^a| =\mathcal{D} |X^a)(I|&=0, \\
    \mathcal{D} |W^a)(W^b| &= \delta^{ab}\mu.
\end{align}\end{subequations}
\end{prop}
\begin{proof}
The key idea behind this proof is to use the group theoretic structure of the random unitaries $D_x$ (for simplicity in what follows, we drop the $x$ subscript, since $\mathcal{D}$ is simply a tensor product of the channel defined via (\ref{eq:depolarize}) on every site anyway).  By definition in our protocol, \begin{equation}
    \mathcal{D} := \frac{1}{48} \sum_{D \in G} \left[ \left(\mathrm{Conj}_D\right)(\cdot)\left(\mathrm{Conj}_D\right)^\dagger \right],
\end{equation}
as 48 is the number of group elements in $G$, defined in (\ref{eq:groupG}).  By construction, $\mathcal{D}$ takes any $4\times 4$ matrix $M$ acting on the space of operators on a qubit, defined in (\ref{eq:BR13}), and projects it onto $G$-invariant maps.   

To find these irreducible representations, we first observe that $G$ is isomorphic to the ``double cover" of the group of rotations which leave invariant a three-dimensional cube, which is in turn isomorphic to $\mathrm{S}_4 \times \mathbb{Z}_2$ (a subgroup of SU(2)). By standard representation-theoretic computation, we find that $\mathbb{R}^{1+3}$ decomposes into two irreducible representations of $G$: \begin{equation}
    \mathrm{span}(I,X,Y,Z) = \mathrm{span}(I) \oplus \mathrm{span}(X,Y,Z)
\end{equation}
By Schur's Lemma, the only $G$-invariant matrices obeying $\mathcal{D}M = M$ are therefore of the form \begin{equation}
    M = a |I)(I| + b |X)(X| + b|Y)(Y| + b|Z)(Z| = a\mathcal{I} + 3b \mu .
\end{equation}
We then obtain (\ref{eq:depolarize}) by using the fact that $\mathcal{D}$ is a  probability-weighted linear sum of unitary operators, and therefore is completely positive and trace preserving.  This condition fixes $a=1$ and $b=\frac{1}{3}$.
\end{proof}

Combining this proposition with (\ref{eq:induction2}), we immediately find:
\begin{cor} \label{cor54}
For any super-operator $\rho$, \begin{equation}
    \mathcal{D} \rho = \sum_{S\subset\Lambda} \rho_S \bigotimes_{i\in S}\mu_i \bigotimes_{j\in S^c}I_j.
\end{equation}
In our random protocol, (\ref{eq:Vqmu}) holds, even in intermediate protocol steps, after any application of $\mathcal{D}$.
\end{cor}

\emph{Step 3:} Corollary \ref{cor54} shows that the only information we need to keep track of, after depolarizing, is the probability $p_S$ that our growing operator is supported on the subset $S$.    We define the number super-operators:
\begin{subequations}\begin{align}
    \hat{N} &= \sum_{x\in\Lambda} \hat{N}_x, \\  \hat{N}_x&:=\big(|X)(X| + |Y)(Y| + |Z)(Z|\big)_x = 3 \mu_x.
\end{align}\end{subequations}
However, due to Corollary \ref{cor54}, we can actually think of each of these as \emph{classical random variables}.  In particular, after applying (averaged) super-channel $\mathcal{DU}$ (for any $\mathcal{U}$),  the probability distribution $\hat{\BP}$ on the variables $\hat{N}_x \in \lbrace 0,1\rbrace$ can be defined (with a slight abuse of notation) via \begin{align}
    \hat{\BP}[\hat{N}_{x_1}=\hat{N}_{x_2}&=\cdots = \hat{N}_{x_m}=1, \hat{N}_{y_1}=\cdots = \hat{N}_{y_n}=0] \notag \\
    &:= \mathrm{Tr}\left[ \hat{N}_{x_1}\cdots \hat{N}_{x_m}(1-\hat{N}_{y_1})\cdots (1-\hat{N}_{y_n}) \mathcal{DU} |X_0)(X_0| \right].
\end{align}
In words, we consider the basis which all $\hat{N}$ are simultaneously digonalized.  The classical probability distribution $\hat{\BP}$ counts the weight of $\mathcal{DU} |X_0)(X_0|$ that satisfies the boolean quantifiers represented by the product of $\hat N$s and $(1-\hat N)$s above: namely, the sum of all diagonal elements of $\mathcal{DU} |X_0)(X_0|$ compatible with $\hat N_{x_1}=1$, etc.  The well-posedness of $\hat{\BP}$ is guaranteed by the fact that $\mathcal{DU}$ is completely positive and trace preserving, and the observables $N$ are commuting.  Knowing the classical probabilities $p_S$ (corresponding to the probability that $\hat N_x = 1$ if and only if $x\in S$) is enough to know $\mathcal{DU}|X_0)(X_0|$.

Armed with this knowledge, we are ready to lay out the foundations for the remainder of our inductive proof.  The induction hypothesis we begin with is that for any set $S\subset C_0 \in \mathcal{B}_{q-1}$, with $C_0$ the $(q-1)$-scale cube containing the origin, \begin{align}
         \hat{\BP}\left(\hat{N}\CM _{q-1} \bigotimes_{i\in S} \mu_i  \ge \lambda_{1,q-1} R^d_{q-1}\right) \ge \eta_{1,q-1}. \label{eq:induchyp}
    \end{align}
Here $0 < \lambda_{1,q-1},\eta_{1,q-1}<1$ are constants that we will obtain a little later; in particular though, $\eta_{1,q-1}$ will be interpreted as a lower bound on the \emph{success probability} of $\mathcal{M}_{q-1}$ -- namely, the probability that it grows an operator to have support on $S\subset C_0$ with $S$ containing fraction $\ge \lambda_{1,q-1}$ of all sites in $C_0$.

Next, we condition on the assumption that at least $s:= \lambda_{1,q-1}R_{q-1}^d$ sites are occupied.  This throws away a fraction of at most $1-\eta_{1,q-1}$ of the operator (an amount that we will see is small).  We will then show that if we run the unitary $V_q$ for sufficiently long time,
we will seed more than $3.3\%$ of the $(q-1)-cubes$ in the $q$-cube containing the origin, with probability $\eta_{2,q-1}$:  see Lemma \ref{lem:ZZ-step}.  We note that $1-\eta_{2,q-1}$ decays exponentially with $m^d$ and $s$.

Lastly, we apply $\CM _{q-1}$, which will attempt to grow the operators in $(q-1)$-cubes $C_{q-1}(\vec{k})$ that we seeded above into large Pauli strings. Conditioned on there are at least $n_s$ seeded blocks $C_{q-1}$, we will show in Proposition \ref{prop:expand_after_seeding} that with probability $\eta_{3,q-1}$, a rather large number of sites (to be quantified later) will be occupied (when considering all $(q-1)$-cubes together).  This proves the inductive hypothesis.  We summarize the way that operators grow throughout this process in Figure \ref{fig:full_protocol}.

\begin{figure}[t]
    \centering
    \includegraphics[width=0.9\textwidth]{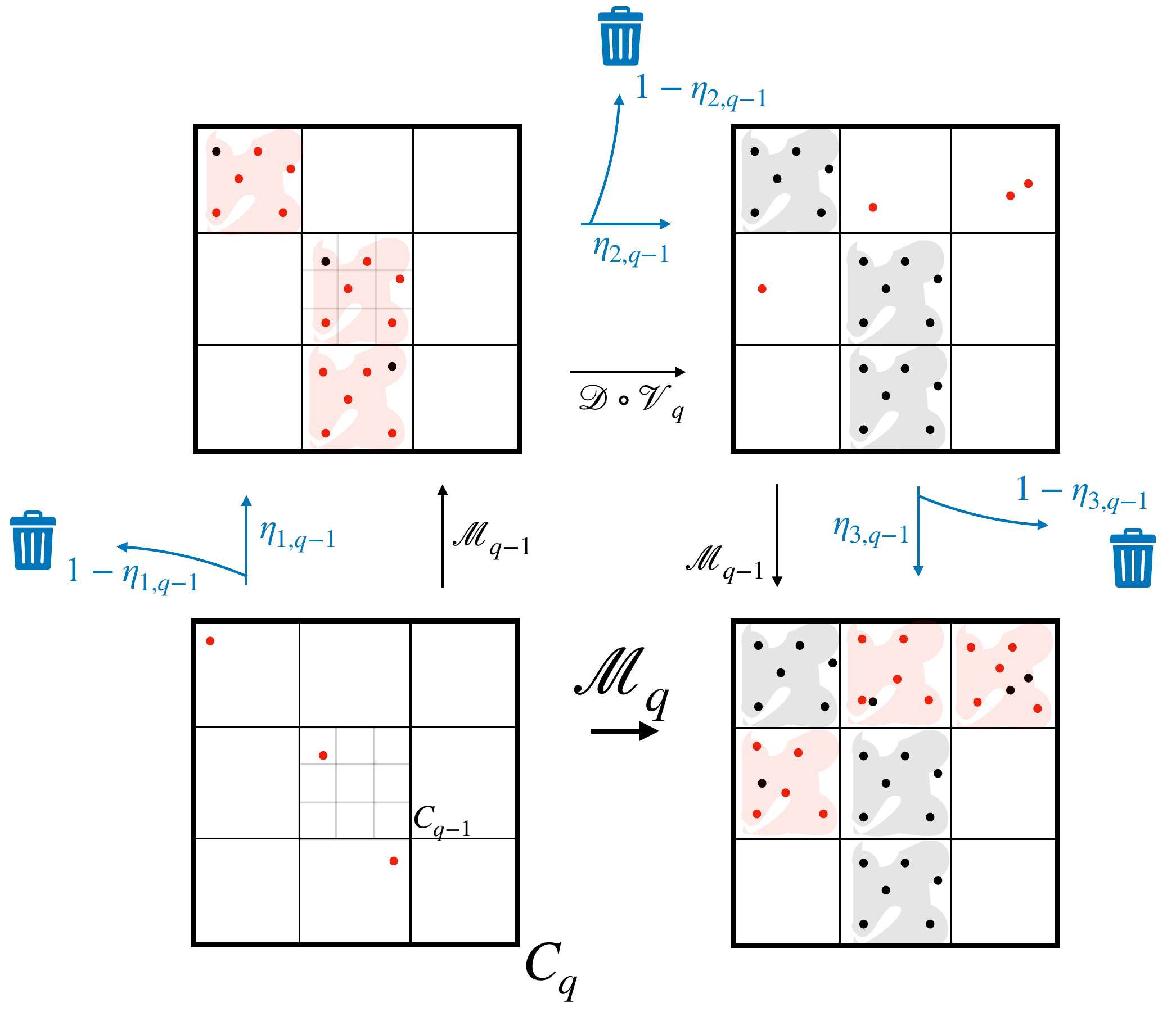}
    \caption{The rigorous operator growth protocol, in which we no longer guarantee full occupancy of all sites.  Instead, we only get to keep a fraction of occupancy at each recursion, albeit with high probability.  Using the central limit theorem, the failure probabilities $1-\eta$ is suppressed by rate of expansion $e^{-\CO(m)}$.  The choice of $m\sim \exp(\sqrt{\ln r})$ is large enough to render these probabilities mild, and leave us with (\ref{eq:bigpropeqn}).
    }
    \label{fig:full_protocol}
\end{figure}

\emph{Step 4.} We now analyze the $V_q$ operator growth step of our protocol in detail.  The key observation is that, upon averaging over all possible $V_q$, the superchannel $\mathcal{V}_q$ takes an arbitrary operator of sufficient size in any $(q-1)$-cube $C_{q-1}\subset C_q$, and ``seeds" new $\mu$s in a finite fraction of the remaining $(q-1)$-cubes $C^\prime_{q-1} \subset C_q$ which are contained within the $q$-cube $C_q$ (see Figure \ref{fig:full_protocol}).   More precisely, we have:

\begin{lem}\label{lem:ZZ-step}
Let $C_q \in \mathcal{B}_q$.  For any subset $S\subseteq C_q$ of size $|S|\ge s$, set the time
\begin{align}
    \tau_q: =  \frac{R^\alpha_q}{\sqrt{2 s R^d_{q-1}}}. \label{eq:tauqdef}
\end{align}
Define the superoperator \begin{equation}
    \rho^\prime = \mathcal{V}_q \left( \mu^{\otimes S} \mathcal{I}^{\otimes (C_q-S)}\right). \label{eq:rhoprime}
\end{equation} 
Let $C \in \mathcal{B}_{q-1}$ denotes one of the $(q-1)$-cubes contained within $C_q$, and define the classical random variables \begin{equation}
    \hat{y}_C := \indicator\left(\sum_{x\in C} \hat{N}_x > 0 \right).
    \end{equation}
Then 
\begin{align}
1-\eta_{2,q-1}: =\hat{\BP}_{\rho^\prime}\left[  \sum^{m^d}_{C=1} \hat{y}_C \le \frac{m^d}{30} \right] \le \exp\left(-\frac{m^d}{120} \right)+ \exp\left( -\frac{s}{72} \right) \label{eq:eta2bound}
\end{align}
where $\hat{\BP}_{\rho^\prime}$ denotes the probability distribution obtained from (\ref{eq:rhoprime}).  For notational simplicity, we have labeled the $m^d$ $(q-1)$-cubes in $C_q$ as $C=1,\ldots, m^d$.
\end{lem}
\begin{proof}
First, we re-write $\mathcal{V}_q$ in a more elegant way, using the fact that all $ZZ$ terms in $H^{ZZ}_q$ commute with each other: \begin{align}
    \mathcal{V}_q &= \BE_{J}\left[ \exp\left[ \frac{\tau_q}{R_q^\alpha} \sum_{\lbrace i,j\rbrace \subset C_q} J_{ij} \mathrm{Adj}_{\mathrm{i}Z_i Z_j} \right] \right](\cdot)\left[ \exp\left[ \frac{\tau_q}{R_q^\alpha} \sum_{\lbrace i,j\rbrace \subset C_q} J_{ij} \mathrm{Adj}_{\mathrm{i}Z_i Z_j} \right] \right]^\dagger \notag \\
    &=\BE_J \prod_{\lbrace i,j\rbrace \subset C_q}\ \ \left[ \exp\left[ \frac{\tau_q}{R_q^\alpha} J_{ij}\mathrm{Adj}_{\mathrm{i}Z_i Z_j}\right]\right](\cdot)\left[ \exp\left[ \frac{\tau_q}{R_q^\alpha} J_{ij}\mathrm{Adj}_{\mathrm{i}Z_i Z_j}\right]\right]^\dagger, \label{eq:Vq55}
\end{align}
Since $Z$ terms commute with $Z$s, but not with $X$s or $Y$s, let us define \begin{equation}
    \nu = \frac{|X)(X| + |Y)(Y|}{2}, \;\;\;\; \mathcal{Z} = |Z)(Z|
\end{equation}
such that \begin{equation}
\mu = \frac{2}{3}\nu + \frac{1}{3}\mathcal{Z}. \label{eq:munu}
\end{equation}

From (\ref{eq:Vq55}), we can first calculate what a single $Z_iZ_j$ coupling does to $\mu_i$ or $\mu_i\mu_j$; the application of $\mathcal{V}_q$ will then correspond to repeating this procedure on all pairs of sites in the cube $C_q$.   It is straightforward to show that (temporarily denoting $\theta_{ij} := J_{ij}\tau_q R_q^{-\alpha}$)\begin{subequations}\begin{align}
    \exp[\theta \mathrm{Adj}_{\mathrm{i}Z_iZ_j}]\ \mathcal{I}_i \mathcal{I}_j\ \exp[\theta \mathrm{Adj}_{\mathrm{i}Z_iZ_j}]^\dagger&= \mathcal{I}_i \mathcal{I}_j, \\
    \exp[\theta \mathrm{Adj}_{\mathrm{i}Z_iZ_j}]\ \nu_i \nu_j\ \exp[\theta \mathrm{Adj}_{\mathrm{i}Z_iZ_j}]^\dagger  &= \nu_i \nu_j.
\end{align}\end{subequations}
Moreover, since \begin{equation}
    \exp[\theta \mathrm{Adj}_{\mathrm{i}Z_iZ_j}] \left(\begin{array}{c} |X_iI_j) \\ |Y_iZ_j)  \end{array}\right) = \left(\begin{array}{cc} \cos(2\theta) &\ \sin(2\theta) \\ -\sin(2\theta) &\ \cos(2\theta) \end{array}\right) \left(\begin{array}{c} |X_iI_j) \\ |Y_iZ_j)  \end{array}\right),
\end{equation}
we conclude that \begin{equation}
   \mathbb{E}\left[ \exp[\theta \mathrm{Adj}_{\mathrm{i}Z_iZ_j}] \left(\begin{array}{c} |X_iI_j)(X_i I_j| \\ |Y_iZ_j)(Y_i Z_j|  \end{array}\right)  \exp[\theta \mathrm{Adj}_{\mathrm{i}Z_iZ_j}]^\dagger\right]= \left(\begin{array}{cc} 1-p_1 &\ p_1 \\ p_1 &\ 1-p_1 \end{array}\right) \left(\begin{array}{c} |X_iI_j)(X_i I_j| \\ |Y_iZ_j)(Y_i Z_j|  \end{array}\right) \label{eq:p155}
\end{equation}
where \begin{equation}
    p_1 = \mathbb{E}[\cos^2(2\theta_{ij})] = \frac{1}{2}-\frac{1}{2}\mathbb{E}\left[\mathrm{e}^{4\mathrm{i}J_{ij}\tau_q/R_q^\alpha}\right] = \frac{1}{2}\left[1-F\left(\frac{4\tau_q}{R_q^\alpha}\right)\right]
\end{equation}
where \begin{equation}
    F(x) := \frac{\sin x}{x}.
\end{equation}
The following Taylor expansion estimate for $F(x)$ will prove useful: \begin{equation}
    1- \frac{x^2}{6}\le F(x) \le 1- \frac{x^2}{6} + \frac{x^4}{120}. \label{eq:Ftaylor}
\end{equation}
Combining (\ref{eq:p155}) and (\ref{eq:Ftaylor}), we obtain \begin{equation}
    \mathbb{E}\left[ \exp[\theta \mathrm{Adj}_{\mathrm{i}Z_iZ_j}] \left(\begin{array}{c} \nu_i \mathcal{I}_j \\ \nu_i \mathcal{Z}_j \end{array}\right) \exp[\theta \mathrm{Adj}_{\mathrm{i}Z_iZ_j}]^\dagger\right] = \left(\begin{array}{cc} 1-p_1 &\ p_1 \\ p_1 &\ 1-p_1   \end{array}\right) \left(\begin{array}{c} \nu_i \mathcal{I}_j \\ \nu_i \mathcal{Z}_j \end{array}\right)
\end{equation}
where, upon plugging in for $\theta_{ij}$, we obtain
\begin{equation}
    \frac{4}{3}\frac{\tau^2_q}{R^{2\alpha}_q} > p_1 > \frac{4}{3}\frac{\tau^2_q}{R^{2\alpha}_q}-\frac{16}{15}\frac{\tau^4_q}{R^{4\alpha}_q}.  \label{eq:Fbound}
\end{equation}

An important consequence of all of these identities is that if we write out an initial operator density in the form \begin{equation}
    \hat{\rho} = \prod_{i\in S_\nu} \nu_i \prod_{j \in S_Z} \mathcal{Z}_j \prod_{k\in C_q - S} \mathcal{I}_k,
\end{equation}
where $S_\nu$ and $S_Z$ form a partition of $S$, then this tensor product structure \emph{is preserved under the average of the time evolution}.   This is because $\nu_i$s are invariant, while $\mathcal{Z}$s and $\mathcal{I}$s convert between each other in a Markovian fashion independently on each site.  Using the independence of the $J_{ij}$s, we see that if $\ell=|S_\nu|$, then \begin{equation}
    \mathcal{V}_q \hat{\rho} =  \prod_{i\in S_\nu} \nu_i \prod_{j \in S_z} \left[ (1-p_\ell) \mathcal{Z}_j + p_\ell \mathcal{I}_j\right] \prod_{k\in C_q - S} \left[ (1-p_\ell) \mathcal{I}_k + p_\ell \mathcal{Z}_k\right], \label{eq:Vqrho}
\end{equation}
where we define $p_\ell$ via \begin{equation}
    \left(\begin{array}{cc} 1-p_1 &\ p_1 \\ p_1 &\ 1-p_1   \end{array}\right)^\ell = \left(\begin{array}{cc} 1-p_\ell &\ p_\ell \\ p_\ell &\ 1-p_\ell   \end{array}\right).
\end{equation}
More concretely, we may bound $p_\ell$ as follows:
\begin{align}
    p_\ell &= \sum_{j=1,3,5,\ldots}^\ell \left(\begin{array}{c} \ell \\ j \end{array}\right) p_1^j (1-p_1)^{\ell-j} =
    \frac{1}{2} \left[ (1-p_1+p_1)^\ell- \left(1-p_1-p_1\right)^\ell\right] \ge \frac{1- \mathrm{e}^{-2p_1\ell}}{2}. \label{eq:pl}
\end{align}


Now, of the $m^d$ cubes $C_{q-1}(\vec{k})\subset C_q$, we must count how many of them contain either a $\nu$ or a $\mathcal{Z}$ on at least one site, after applying $\mathbb{E}[\mathcal{V}_q] \hat{\rho}$.  
Consider Bernoulli random variables $\hat{x}_j \in \lbrace 0,1\rbrace$ on sites $j\in C_q$, with \begin{equation}
    \hat{\BP}[\hat{x}_j=1] := \left\lbrace\begin{array}{ll} 1 &j\  \in S_\nu \\ 1-p_\ell &\ j \in S_z \\ p_\ell &\ \text{otherwise} \end{array}\right..
\end{equation}
The interpretation of $\hat{x}_j$ is the probability that we find a non-identity operator on site $j$ in $\mathbb{E}[\mathcal{V}_q] \hat{\rho}$.  Since $p_\ell \le 1/2$, we can easily see that for any $(q-1)$-cube $C \subset C_q$, \begin{equation}
    \hat{\BP}\left[ \sum_{k \in C} \hat{x}_k = 0 \right] = \hat{\BP}\left[ \hat{y}_C= 0 \right]  \le \left(1-p_\ell\right)^{|Q|} := 1-p_* \label{eq:pstar}
\end{equation}
where $Q$ is any subset of $C_q$, which will be taken as $C_{q-1}(\vec{k})\subset C_q$. This inequality follows from the independence of $\hat{x}_k$, regardless of the set $Q$. 
We can now bound the probability that at least $\lambda_{2,q-1} m_q^d$ of the cubes $C_{q-1}$ have at least one operator in them using the standard Chernoff bounds:  
\begin{prop}[Chernoff bounds]\label{fact:chernoff}
Let $A$ be a discrete set.  For $i\in A$, let $x_i \in \lbrace 0, 1\rbrace$ be independent Bernoulli random variables.  If \begin{equation}
    S := \sum_{i\in A} x_i,
\end{equation}
then
 \begin{equation}
        \BP(S \le (1-\delta) \mathbb{E}[S] ) \le \exp\left(-\frac{1}{2}\delta^2\mathbb{E}[S]\right).
\end{equation}
\end{prop}

In our calculation, the Bernoulli random variables of interest are $\hat{y}_C$, which obey $\mathbb{E}[\hat y_C] \ge p_*$.   Choosing \begin{equation}
    \lambda_{2,q-1} := \frac{p_*}{2},
\end{equation}
we arrive at 
\begin{equation}
    \hat{\BP}\left[  \sum_C \hat{y}_C \le \frac{p_*m^d}{2} \right] \le \exp\left(-\frac{p_*m^d}{8} \right). \label{eq:pstar8}
\end{equation} 

To conclude the proof, note that the above calculation was based on the number of $\nu$s. By our initial assumption, the initial operator will have $\ge s$ $\mu$s, which can be \emph{anywhere} in cube $C_q$.   Breaking up $\mu$ into $\nu$ and $\CZ$(\ref{eq:munu}), we see that our initial operator is binomial distributed: on each of (at least) $s$ sites, the probability of $\nu$ is $2/3$. Using Chernoff bounds on an initial operator of $s$ $\mu$, the probability of having $\ell \ge s/2$ is extremely high when $s$ is large:
\begin{equation}
    \hat{\BP}\left(\ell \ge \frac{s}{2}\right) \ge 1 - \exp\left( -\frac{s}{72} \right). \label{eq:ls2bound}
\end{equation}
This inequality holds regardless of the number of sites in the set $S$ on which our operator is initially supported, so long as $|S|\ge s$.

Now, we proceed as follows.  Using the simple fact that in classical probability theory, for any two events $A$ and $B$, $\mathbb{P}(A) \ge \mathbb{P}(A \text{ and } B)$, we will lower bound the probability $\eta_{2,q-1}$ by calculating \begin{align}
    \hat{\BP}_{\rho^\prime}\left[  \sum^{m^d}_{C=1} \hat{y}_C \le \frac{m^d}{30} \right] &\ge \hat{\BP}_{\rho^\prime}\left[  \sum^{m^d}_{C=1} \hat{y}_C \le \frac{m^d}{30} \text{ and } \ell \ge \frac{s}{2}\right] \notag \\
    &\ge \hat{\BP}_{\rho^\prime}\left[\left.  \sum^{m^d}_{C=1} \hat{y}_C \le \frac{m^d}{30} \right| \ell \ge \frac{s}{2}\right] \hat{\BP}_{\rho^\prime}\left[\ell\ge\frac{s}{2}\right]. \label{eq:55cond}
\end{align} 
To bound the conditional probability above, we start with (\ref{eq:pl}) and (\ref{eq:pstar}):
\begin{align}
    p_* &= 1-(1-p_\ell)^{R_{q-1}^d} \ge 1 - \exp\left(-p_\ell R_{q-1}^d\right) \ge 1 - \exp\left(-\frac{1-\mathrm{e}^{-p_1s}}{2} R_{q-1}^d\right)  \label{eq:pstarbound1}
\end{align}
Note that we have used that $\ell \ge s/2$.  Now, observe that for $x<1$, $\mathrm{e}^{-x} \le 1 - \frac{x}{2}$.  Since \begin{align}
    p_1s\le  \frac{4s}{3}\frac{\tau^2_q}{R^{2\alpha}_q} \le \frac{2}{3R^{d}_{q-1}} \le 1, \label{eq:p1sbound}
\end{align}
we may further simplify (\ref{eq:pstarbound1}) to \begin{equation}
    p_* \ge 1-\exp\left(-\frac{p_1s R_{q-1}^d}{4}\right) \ge \frac{p_1 sR_{q-1}^d}{6},
\end{equation}
noting that (\ref{eq:p1sbound}) implies the argument of exp above is $\le 1/6$, in which case $1-\mathrm{e}^{-x}\ge \frac{2}{3}x$.   Now using (\ref{eq:tauqdef}) and (\ref{eq:Fbound}), we find  \begin{equation}
    p_* \ge \frac{sR_{q-1}^d}{6} \left(\frac{2}{3sR_q^{d-1}} - \frac{4}{15s^2 R_q^{2(d-1)}} \right) \ge \frac{sR_{q-1}^d}{6} \frac{2}{3sR_q^{d-1}} \left(1-\frac{2}{5sR_q^{d-1}}\right) \ge \frac{1}{15}. \label{eq:pstar15}
\end{equation}
Combining (\ref{eq:pstar8}) and (\ref{eq:pstar15}), we find that \begin{equation}
    \hat{\BP}_{\rho^\prime}\left[\left.  \sum^{m^d}_{C=1} \hat{y}_C \le \frac{m^d}{30} \right| \ell \ge \frac{s}{2}\right] \ge 1 - \exp\left(-\frac{m^d}{120}\right). \label{eq:yCbound}
\end{equation} 
Combining (\ref{eq:ls2bound}), (\ref{eq:55cond}) and (\ref{eq:yCbound}), we obtain (\ref{eq:eta2bound}).
\end{proof}

\emph{Step 5.}  Now that we know at least $1/30$th of the $(q-1)$-cubes in a $q$-cube are seeded after the first application of $\mathcal{V}_q$, we must now ask what happens after applying $\mathcal{M}_{q-1}$ again (recall Figure \ref{fig:full_protocol}).  The answer is provided by the following proposition:

\begin{prop}\label{prop:expand_after_seeding}
Consider a tensor product superoperator of the form \begin{equation}
    \rho = \mu^{\otimes S} \mathcal{I}^{\otimes S^{\mathrm{c}}},
\end{equation}
where $S\subset \Lambda$ is finite.  Then \begin{align}
    1-\eta_{3,q-1} := \hat{\BP}_\rho\left[\left.  \hat{N} \le \frac{w\eta_{1,q-1} \lambda_{1,q-1} R^d_{q-1}}{2} \right| \sum_{C\in \mathcal{B}_{q-1}} \hat y_C = w\right] \le \exp\left(-\frac{w\eta_{1,q-1}}{8}\right) \label{eq:step5}
\end{align}
\end{prop}
\begin{proof}
By our inductive hypothesis (\ref{eq:induchyp}), for any $C\in \mathcal{B}_{q-1}$,\begin{equation}
    \hat{\BP}\left[ \left.\sum_{x\in C}\hat{N}_x \mathcal{M}_{q-1}\rho \ge \lambda_{1,q-1}R_{q-1}^d \right|\mathrm{Tr}\left(\hat{y}_C\rho\right)= 1\right] \ge \eta_{1,q-1}.
\end{equation}
So, letting $\hat{x}_C$ denote a Bernoulli random variable denoting whether or not the criterion above is satisfied, we observe that $\mathbb{E}[\hat{x}_C] \ge \eta_{1,q-1}$.  Moreover, $\hat{x}_C$ form independent Bernoulli random variables for each cube $C$.  So again, we may use the Chernoff bounds to show that \begin{equation}
    \hat{\BP}\left[ \sum_{C\in\mathcal{B}_{q-1}} \hat{x}_C \le \frac{w\eta_{1,q-1}}{2}\right] \le \exp\left(-\frac{w\eta_{1,q-1}}{8}\right). \label{eq:hatxC}
\end{equation}
Since we are guaranteed that the event in (\ref{eq:step5}) does \emph{not} occur if the event in (\ref{eq:hatxC}) does \emph{not} occur (since for each $\hat{x}_C\ge 1$, we get a contribution of at least $\lambda_{1,q-1}R_{q-1}^d$ to $\hat N$), (\ref{eq:hatxC}) implies (\ref{eq:step5}).
\end{proof}

\emph{Step 6.} We now must combine the results from the previous two steps to prove that we may choose $\eta_{1,q}$ such that (\ref{eq:induchyp}) continues to hold at scale $q$.   As in Figure \ref{fig:full_protocol}, we can conclude from our discussion above that if we choose \begin{equation}
    \lambda_{1,q} :=\frac{\eta_{1,q-1} \lambda_{1,q-1} }{60}, \label{eq:lambdachoice}
\end{equation}
then the success probability at scale $q$ is given by
\begin{align}
    \eta_{1,q}:=\hat{\BP}\left[ \mathrm{Tr}\left( \hat{N}\CM [\bigotimes_{i\in S\subset C_q} \mu_i] \right)\ge \lambda_{1,q} R^d_{q} \right] \ge \eta_{1,q-1}\eta_{2,q-1}\eta_{3,q-1},
\end{align}
where $\eta_{1,q-1}$ is the probability that the first $\mathcal{M}_{q-1}$ is successful (namely, the induction event in (\ref{eq:induchyp}) has occurred), $\eta_{2,q-1}$ is the probability that the first $\mathcal{DV}_q$ is successful (given in Lemma \ref{lem:ZZ-step}), and $\eta_{3,q-1}$ is the probability that the second $\mathcal{M}_{q-1}$ is successful (given in Proposition \ref{prop:expand_after_seeding}).

To start off this recursive relation, we first discuss what happens at scale $q=1$.  This scale is somewhat special, since the operator starts off with probability 1 being a Pauli $X$ at the origin.  Because it is a Pauli $X$, the bound in Lemma \ref{lem:ZZ-step} can be simplified (\emph{only} for this very first application of $\mathcal{V}_1$, before any depolarizing channel!).  It is in fact simplest to convert $|X_0)(X_0|$, the initial super-operator, into $\nu_0$ (this simplifies some equations, but does not change the protocol's performance).  Using (\ref{eq:Vqrho}), we see that \begin{equation}
    \mathcal{V}_1|X_0)(X_0| = \nu_0 \prod_{x \in C_1 - 0} \left((1-p_1)\mathcal{I}_x + p_1\mathcal{Z}_x\right).
\end{equation}
where $C_1$ denotes the 1-cube containing the origin, and, using (\ref{eq:Fbound}), \begin{equation}
    p_1 > \frac{2}{3} - \frac{4}{15} = \frac{6}{15}.
\end{equation}
Using the Chernoff bounds, we can easily see that \begin{equation}
    1-\eta_{2,0} = \hat{\BP}\left[ \mathrm{Tr}\left(\hat N \mathcal{M}_1 \nu_0\right) \le \frac{m^d}{30}\right] \le \exp\left(-\frac{m^d}{120}\right). \label{eq:eta20orig}
\end{equation}
Note also that $\eta_{1,0}=\eta_{3,0}=1$, since there is no $\mathcal{M}_0$, so \begin{equation}
    \eta_{1,1} = \eta_{2,0} \ge 1- \exp\left(-\frac{m^d}{120}\right). \label{eq:eta11}
\end{equation}
This corresponds to the base case of our inductive proof.

Given this base case, the following lemma demonstrates that $\eta_{1,q}$ is indeed (for all $q\le q_*$) finite.
\begin{lem}
For large enough $r$, the recursion relation
\begin{align}
     \eta_{1,q} &= \eta_{1,q-1}\eta_{2,q-1}\eta_{3,q-1}
\end{align}
admits a solution obeying, for all $1\le q\le q_*$,
\begin{align}
    \eta_{1,q} \ge \frac{1}{2}. \label{eq:eta12}
\end{align} 
\end{lem}
\begin{proof}
The key idea of this proof is that $m$ is (perhaps surprisingly) sufficiently large so that $\eta_{2,k}$ and $\eta_{3,k}$ are so close to 1 that the repeated multiplication of $\eta$ probabilities above converges to a non-zero result (for $q\le q_*$).  To see this concretely, let us use Lemma \ref{lem:ZZ-step} and Proposition \ref{prop:expand_after_seeding} to note that \begin{equation}
    \alpha_{q-1} := \eta_{2,q-1}\eta_{3,q-1} \ge \left(1-\exp\left(-\frac{\lambda_{1,q-1}R_{q-1}^d}{72}\right) - \exp\left(-\frac{m^d}{120} \right)\right)\left(1-\exp\left(-\frac{m^d\eta_{1,q-1}}{240}\right)\right).
\end{equation}
We wish to analyze the nonlinear recursion relation \begin{equation}
    \eta_{1,q} = \eta_{1,q-1} \alpha_{q-1}.
\end{equation}
Happily, to demonstrate the lemma, we can assume that (\ref{eq:eta12}) holds when evaluating $\alpha_{q-1}$.  The reason for this is that $\alpha_{q-1}$ monotonically increases with $\eta_{1,q-1}$; hence setting $\eta_{1,q-1}=1/2$ when evaluating $\alpha_{q-1}$ gives us a lower bound on $\alpha_{q-1}$, and in turn on $\eta_{1,q}$.  From (\ref{eq:lambdachoice}) and (\ref{eq:eta20orig}), we can bound \begin{equation}
    \lambda_{1,q} \ge \lambda_{1,1} 120^{1-q} = \frac{4}{120^q}.
\end{equation}
We thus find that for $q>1$,\begin{align}
    \alpha_{q-1} &\ge 1 - \exp\left(-\frac{m^d}{480}\right) - \exp\left(-\frac{m^d}{120}\right) - \exp\left(-\frac{1}{18}\left(\frac{m^d}{120}\right)^{q-1}\right) \notag \\
    &\ge 1 - \exp\left(-\frac{m^d}{480}\right) - \exp\left(-\frac{m^d}{120}\right) - \exp\left(-\frac{m^d}{2160}\right).
\end{align}
We conclude that for all $q \le q_*$, \begin{equation}
    \eta_{1,q} \ge \left[1 - \exp\left(-\frac{m^d}{480}\right) - \exp\left(-\frac{m^d}{120}\right) - \exp\left(-\frac{m^d}{2160}\right)\right]^{q_*}. \label{eq:eta1qbound}
\end{equation}

Let us now show that (\ref{eq:eta1qbound}) is compatible with (\ref{eq:eta12}).
To do this, we simply recall the definitions of $m \sim \exp(\sqrt{\ln r})$ in (\ref{eq:mdef}), and $q_* \sim \sqrt{\ln r}$ in (\ref{eq:qstardef}).  In fact, since for sufficiently (and not very) large $r$, $\sqrt{\ln r} > \ln(\ln r)$,\begin{equation}
    \exp\left(-m^d\right) \le \exp\left(-\mathrm{e}^{d\sqrt{\ln r}}\right) \le \exp\left(-\mathrm{e}^{d\ln (\ln r)}\right) \le \exp\left(-(\ln r)^d\right) \le \frac{1}{r}.
\end{equation}
Thus we observe that for sufficiently large $m$ (and hence $r$), \begin{align}
    \left[1 - \exp\left(-\frac{m^d}{480}\right) - \exp\left(-\frac{m^d}{120}\right) - \exp\left(-\frac{m^d}{2160}\right)\right]^{q_*} &> 1-3q_* \exp\left(-\frac{m^d}{2160}\right) \notag \\
    &> 1-\frac{3(1+\sqrt{\ln r})}{r^{1/2160}}.
\end{align}
For sufficiently large $r$, this is larger than 1/2.  This ensures (\ref{eq:eta12}).
\end{proof}

The final step in the proof of Proposition \ref{thm5} is very simple.  We have shown that on average, we can take a single Pauli $X_0$ supported on one site, and have half of the operator weight supported on a fraction of sites \begin{equation}
    \hat{N} \ge \lambda_{1,q_*}R_{q_*}^d \ge \left(\frac{m^d}{120}\right)^{q_*} \ge \left(\frac{\mathrm{e}^{d\sqrt{\ln r}}}{120}\right)^{\sqrt{\ln r}} \ge \frac{r^d}{120^{\sqrt{\ln r}}}.
\end{equation}
Setting $r=L^{1+\epsilon/2}$, we find \begin{equation}
    \hat{N} \ge L^d \cdot \frac{L^{d\epsilon/2}}{120^{\sqrt{(1+\epsilon/2)\ln L}}} > L^d \label{eq:NLd}
\end{equation}
for sufficiently large $r$ and $L$.  Combining (\ref{eq:eta12}) with (\ref{eq:NLd}), we obtain that (\ref{eq:bigpropeqn}) holds on average in our ensemble of random unitaries:
\begin{align}
    \BE\left[ \lVert \mathbb{P}_{\ge L} |X_0(t))\rVert_{\mathrm{F}}\right] \ge \frac{1}{2}.
\end{align}
Since  $\lVert\mathbb{P}_{\ge L} |X_0(t))\rVert_{\mathrm{F}}\le 1$, with at least $50\%$ chance a unitary drawn from the ensemble has grown an operator to be large: 
\begin{align}
    \mathbb{P}\left[\lVert \mathbb{P}_{\ge L} |X_0(t))\rVert_{\mathrm{F}} \ge \frac{1}{2}\right] \ge \frac{1}{2}.
\end{align}


\end{appendix}

\bibliography{ref}

\end{document}